\documentclass[12pt]{article}

\usepackage[utf8]{inputenc}
\usepackage{changepage}
\usepackage{color}

\usepackage[colorlinks = true,
            linkcolor = blue,
            urlcolor  = blue,
            citecolor = blue,
            anchorcolor = blue]{hyperref}

\usepackage{amsmath}
\usepackage{amsfonts}
\usepackage{amssymb}
\usepackage{mathtools}
\usepackage{comment}
\usepackage[ruled,vlined]{algorithm2e}
\usepackage[margin=-1cm]{caption}
\usepackage{subcaption}
\usepackage{bbm}
\usepackage{natbib}
\usepackage{tabularx}
\newcolumntype{Y}{>{\centering\arraybackslash}X}
\newcolumntype{L}[1]{>{\raggedright\arraybackslash}p{#1}}
\newcolumntype{C}[1]{>{\centering\arraybackslash}p{#1}}
\newcolumntype{R}[1]{>{\raggedleft\arraybackslash}p{#1}}
\usepackage{array,booktabs}

\let\savebigtimes\bigtimes
\let\bigtimes\relax
\usepackage{mathabx}
\let\bigtimes\savebigtimes

\usepackage[standard]{ntheorem}

\usepackage{thmtools} 

\newcommand{\tcb}{\textcolor{blue}}

\DeclareMathOperator*{\argmax}{arg\,max}
\DeclareMathOperator*{\argmin}{arg\,min}

\DeclareMathOperator*{\Var}{Var}
\DeclareMathOperator{\Tr}{Tr}
\newcommand{\E}{\ensuremath{{\mathbb E}}} 

\newcommand{\R}{\ensuremath{{\mathbb R}}}

\newcommand{\Sample}{\ensuremath{ \mathbf{\hat{S}}  }}
\newcommand{\f}[1]{F\left(#1\right)}
\newcommand{\Orth}{\ensuremath{\mathcal{O}}}
\newcommand{\Vsol}{\ensuremath{\mathbf{\hat{V}}}}
\newcommand{\VT}{\ensuremath{\mathbf{V}}}
\newcommand{\Vset}{\ensuremath{\mathcal{V}}}
\newcommand{\Vit}[1]{\ensuremath{\mathbf{V}^{[#1]}}}
\newcommand{\LQIS}{\ensuremath{\boldsymbol{\Lambda}_{0}}}
\newcommand{\eqiv}{\sim_{\boldsymbol{\Lambda}_{0}}}
\newcommand{\LRQIS}{\ensuremath{\boldsymbol{\Lambda}_{R}}}
\newcommand{\LT}{\ensuremath{\boldsymbol{\Lambda}}}
\newcommand{\EigFVhat}{\ensuremath{\boldsymbol{\hat{\Lambda}}}}
\newcommand{\EigFl}[1]{\ensuremath{\boldsymbol{\hat{\Lambda}}^{[#1]}}}
\newcommand{\Ortheq}{\ensuremath{\mathcal{O}_0}}
\newcommand{\HTyler}{\ensuremath{\mathbf{\hat{H}}_{T}}}
\newcommand{\NL}[1]{\mbox{NL}\left(#1\right)}

\newcommand{\LS}[1]{\mbox{LS}\left(#1\right)}
\newcommand{\diag}{\ensuremath{\mbox{diag}}}

\newcommand{\qed}{\hfill \ensuremath{\Box}}
\newcommand{\sI}{(I)}
\newcommand{\sA}{(A)}
\newcommand{\sF}{(F)}
\newcommand{\sIsig}{(I$'$)}
\newcommand{\sAsig}{(A$'$)}
\newcommand{\sFsig}{(F$'$)}

\newcommand*{\defeq}{\mathrel{\vcenter{\baselineskip0.5ex\lineskiplimit0pt\hbox{\scriptsize.}\hbox{\scriptsize.}}}=}
\newcommand*{\eqdef}{=\mathrel{\vcenter{\baselineskip0.5ex\lineskiplimit0pt\hbox{\scriptsize.}\hbox{\scriptsize.}}}}

\title{R-NL: Covariance Matrix Estimation for Elliptical
Distributions based on Nonlinear Shrinkage}

\author{{\large Simon Hediger}$^{a}$\hspace*{2mm}
\hspace*{2mm}{\large Jeffrey N{\"a}f$^{b}$\footnote{Corresponding author at: PreMeDICaL, Inria-Inserm, Montpellier, France.\newline E-mail address: \href{mailto:jeffrey.naf@inria.fr}{jeffrey.naf@inria.fr}.}}\hspace*{2mm}
\hspace*{2mm}{\large Michael Wolf}$^{a}$
\\[3mm]
$^{a}$\textit{\small Department of Economics, University of Zurich, Switzerland}\\
$^{b}$\textit{\small PreMeDICaL, Inria-Inserm, Montpellier, France}\
}

\begin{document}

\def\spacingset#1{\renewcommand{\baselinestretch}%
{#1}\small\normalsize} \spacingset{1}



\maketitle
\thispagestyle{empty}


\begin{abstract}
We combine Tyler's robust estimator of the dispersion matrix with nonlinear shrinkage. This approach delivers a simple and fast estimator of the dispersion matrix
in elliptical models that is robust against both heavy tails and 
high dimensions. We prove convergence of the iterative part of our
algorithm and demonstrate the favorable performance of the estimator
in a wide range of simulation scenarios. Finally, an empirical
application demonstrates its state-of-the-art performance on real
data.
\end{abstract}


\newpage



\section{Introduction}

Many statistical applications rely on covariance matrix
estimation. Two common challenges are (1)~the presence of heavy tails and (2)~the high-dimensional nature of the data. Both problems lead to suboptimal performance or even inconsistency of the usual sample covariance estimator $\Sample$. Consequently, there is a vast literature on addressing these problems.

Two prominent ways to address (1) are (Maronna's)
\mbox{$M$-estimators} of scatter (\cite{MestimatorofScatter}), as well
as truncation of the sample covariance matrix; for example, see \cite{user_friendly_cov}. There also appear to be two main approaches to solving problem (2). The first is to assume a specific structure on the covariance matrix to reduce the number of parameters. One example of this is the ``spiked covariance model'', as explored e.g., in \cite{spikemodel1, spikemodel2, optimalshrinkageinspikedcovariancemodel}, a second is to assume (approximate) sparsity and to use thresholding estimators (\cite{sparsity1,sparsity2,sparsity3,sparsity4}). We also refer to \cite{user_friendly_cov} who present a range of general estimators under heavy tails and extend to the case $n > p$, by assuming specific structures on the covariance matrix. If one is not willing to assume such structure, a second approach is to leave the eigenvectors of the sample covariance matrix unchanged and to only adapt the eigenvalues. This leads to the class of estimators of \cite{stein:1975,stein:1986}. Linear shrinkage (\cite{Wolflinear}) as well as nonlinear shrinkage developed in \cite{WolfNL, Wolf2015,
  Analytical_Shrinkage, QIS2020} are part of this class. 

One promising line of research to
address both problems at once is to extend (Maronna's) \mbox{$M$-estimators}
of scatter
\citep{MestimatorofScatter} with a form of shrinkage for high
dimensions. This approach is in particular popular with a specific
example of $M$-estimators called ``Tyler's estimator'' (\cite{Tyler}),
which is derived in the context of elliptical distributions. Several papers have studied this approach, using a convex combination
of the base estimator and a target matrix, usually the (scaled)
identity matrix.  We generally refer to such approaches as robust
linear shrinkage estimators. For instance, \cite{possiblecompet_2014,
  M_estimatortheory_2016, Shrinkage_2021, Hall_Cov} combine the linear
shrinkage with Maronna's $M$-estimators, whereas
\cite{firstshrinkage, linearshrinkageinheavytails, yang2014minimum,
  zhang2016automatic} do so with Tyler's estimator.
Since this approach of combining linear shrinkage with a robust estimator
entails choosing a hyperparameter determining the amount of
shrinkage, the second step often consists of deriving some
(asymptotically) optimal parameter that then can be estimated from
data. The approach results in estimation methods that are
generally computationally inexpensive and it also enables 
strong theoretical results on the convergence of the
underlying iterative algorithms.

Despite these advantages, several problems remain. First, the
performance of these robust methods sometimes does not exceed the
performance of the basic linear shrinkage estimator of
\cite{Wolflinear} in heavy-tailed models, except for small
sample sizes $n$ (say $n < 100$). In fact, the theoretical analysis of
\cite{couillet2014large,M_estimatortheory_2016} shows that robust
$M$-estimators using linear shrinkage are asymptotically equivalent to
scaled versions of the linear shrinkage estimator of
\cite{Wolflinear}. Depending on how the data-adaptive hyperparameter
is chosen, the performance can even deteriorate quickly as the tails
get lighter, as we demonstrate in our simulation study in
Section~\ref{sec:mc}. Second, some robust methods cannot handle the case when
the dimension $p$ is larger than the sample size $n$, such as
\cite{Shrinkage_2021}. Third, some methods propose a choice of
hyperparameter(s) through cross-validation, such as
\cite{Shrinkage_2017, convexpenalties}, which can be computationally expensive. In this paper, we address these problems by developing a simple
algorithm based on \emph{nonlinear} shrinkage (\cite{WolfNL, Wolf2015,
  Analytical_Shrinkage, QIS2020}), inspired by the above robust
approaches and the work of \cite{comfortNL}. In essence, the
algorithm applies the quadratic inverse shrinkage (QIS) method of
\cite{QIS2020} to appropriately standardized data, thereby greatly
increasing its finite-sample performance in heavy-tailed models. Thus,
we refer to the new method as ``Robust Nonlinear Shrinkage''
(\mbox{R-NL}); in
particular, we extend the proposal of \cite{comfortNL}
from a parametric model to general elliptical distributions. This
approach includes an iteration over the space of orthogonal matrices,
which we prove converges to a stationary point. We motivate our
approach using properties of elliptical distributions along the lines
of \cite{linearshrinkageinheavytails, zhang2016automatic, Hall_Cov}
and demonstrate the favorable performance of our method in a wide range
of settings. Notably, our approach (i) greatly improves the
performance of (standard) nonlinear shrinkage in heavy-tailed settings;
does
not deteriorate when moving from heavy to Gaussian tails; (iii) can
handle the case $p>n$; and (iv) does not require the choice of a
tuning parameter.

The remainder of the article is organized as follows. Section
\ref{sec: cont} lists our contributions. Section
\ref{motivatingexample} presents an example to motivate our
methodology. Section \ref{sec:method} describes the proposed new
methodology and provides results concerning the convergence of the new
algorithm. Section \ref{sec:mc} showcases the performance of our
method in a simulation study using various settings for both $p <
n$ and $p > n$. Section \ref{sec:empirics} applies our method to
financial data, illustrating the performance of the method on
real data.

\subsection{Contributions}
\label{sec: cont}

To the best of our knowledge, no paper has so far attempted to combine
nonlinear shrinkage of \cite{WolfNL, Wolf2015, Analytical_Shrinkage,
  QIS2020} with Tyler's method. As such, our approach differs markedly
from previous ones. It is partly based on an $M$-estimator
interpretation, but also adds the nonparametric nonlinear shrinkage
approach. A downside of this approach is  that theoretical
convergence results 
are harder to come by. Nonetheless, we are able to show that the
iterative part of our algorithm converges to a stationary point, a
crucial result for the practical usefulness of the algorithm.

Maybe the closest paper to our method is \cite{ourclosestcompetitor}, where the eigenvalues of Tyler's estimator are iteratively shrunken towards predetermined target eigenvalues, with a parameter $\alpha$ determining the shrinkage strength. Through different objectives, they arrive at an algorithm from which the iterative part of our Algorithm \ref{RNL} can be recovered when setting $\alpha= \infty$. Additionally, using the eigenvalues from nonlinear shrinkage as the target eigenvalues, their method presents an alternative way of combining Tyler's estimator with nonlinear shrinkage. Though they did not originally propose this, this was suggested by an anonymous reviewer. However, while there is an overlap in the two algorithms for the corner case of $\alpha=\infty$, they arrive at their Algorithm 1 from a different angle than we do. Consequently, their theoretical results cannot be applied in our analysis. Moreover, they do not suggest how to choose the tuning parameter $\alpha$. In Appendix \ref{sec:further_results}, simulations indicate that when the target eigenvalues are obtained from nonlinear shrinkage, setting $\alpha=\infty$, and thus maximally shrinking towards the nonlinear shrinkage eigenvalues, is usually beneficial. In addition, these simulations show that the updating of eigenvalues we propose after the iterations converged can lead to an additional boost in performance over their method.




Whereas many of the aforementioned robust linear shrinkage papers have
important theoretical results, the empirical examination of their
estimators in simulations and real data applications is often
limited. We attempt to give a more comprehensive empirical overview in
this paper. Contrary to most of the previous papers, we also consider
a comparatively large sample size of $n=300$ in our simulation
study. Compared to 6 competing methods, our new approach displays a superior performance over a wide range of scenarios. We also
provide a Matlab implementation of our method, as well as the code to
replicate all simulations on \url{https://github.com/hedigers/RNL_Code}.


\begin{table*}
\caption{Notation}
\centering
\begin{tabular}{l|l}                                                            
Symbol & Description \\
\hline
$n$ & Sample size\\
$p$ & Dimensionality\\
$\boldsymbol{\Sigma}\defeq \Var(\mathbf{Y})$ & The covariance matrix of the random vector $\mathbf{Y}$.\\
$\Tr(\mathbf{A})$ & Trace of a square matrix $\mathbf{A}$\\
$\|\mathbf{A}\|_F$ & Frobenius norm
                     $\sqrt{\Tr(\mathbf{A}^{\top}\mathbf{A})}$ of a
                     sqaure matrix $\mathbf{A}$\\
$\mathbf{H}$ & dispersion matrix \\
$\Orth$ & the orthogonal group\\
$\Ortheq$ & equivalence class in $\Orth$\\
$\mathbf{U}$ & arbitrary element of $\Orth$\\
$\VT$ & Eigenvectors of $\mathbf{H}=\VT \LT \VT^{\top}$\\
$\Vit{\ell}$ & $\ell$th iteration of the algorithm\\
$\Vsol$ & critical point/solution/estimate\\
$\Vset$ & subset of critical points of $\Orth$\\
$\LT$ & True ordered eigenvalues of $\mathbf{H}$, up to scaling\\
$\LQIS$ & Initial (shrunken) estimate of $\LT$\\
$\LRQIS$ & Final \mbox{R-NL} (shrunken) estimate of $\LT$\\
$\EigFVhat$ & Eigenvalues of $\f{\Vsol}$\\
$\EigFl{\ell+1}$ & Eigenvalues of $\f{\Vit{\ell}}$\\
$\diag()$ & Transforms a vector $\mathbf{a} \in \R^p$ into an $p \times p$ diagonal matrix $\diag(\mathbf{a})$   
\end{tabular}
\label{tab:notation}
\end{table*}

\section{Motivational Example}\label{motivatingexample}

For a collection of $n$ independent and identically distributed
(i.i.d.) random vectors with values in $\R^p$, let
$\Vsol=\begin{pmatrix} \mathbf{\hat{v}}_{1}, \ldots,
  \mathbf{\hat{v}}_{p} \end{pmatrix}$ be the matrix of eigenvectors of
the sample covariance matrix $\Sample$. Nonlinear shrinkage, just as
the linear shrinkage of \cite{Wolflinear}, only changes the eigenvalues of
the sample covariance matrix, while keeping the eigenvectors
$\Vsol$. That is, nonlinear shrinkage is also in the class of estimators of
the form $\Vsol \Delta \Vsol^{\top}$, with $\Delta$ diagonal, a class that
goes back to \cite{stein:1975,stein:1986}.
It is well known that
\begin{align*}
    &\argmin_{\Delta \text{ diagonal }} \| \boldsymbol{\Sigma} - \Vsol
  \Delta \Vsol^{\top} \|_{F} = \diag \Bigl ( \begin{pmatrix} \delta_1
    &  \ldots& \delta_N \end{pmatrix}^{\top} \Bigr ) \\ &\mbox{with}
               \quad
    \delta_j:=\mathbf{\hat{v}}_{j}^{\top}
\boldsymbol{\Sigma}\mathbf{\hat{v}}_{j}~;        
\end{align*}
for example, see \cite[Section 3.1]{ledoit:wolf:power}. Nonlinear shrinkage takes the sample covariance
matrix $\Sample$ as an input and outputs a shrunken estimate of
$\boldsymbol{\Sigma}$ of the form $\Vsol \LQIS \Vsol^{\top}$, where
$\LQIS=\mbox{diag}(\hat{\delta}_1, \ldots, \hat{\delta}_N)$ is a
diagonal matrix. Although there are different schemes to come with
estimates $\{\hat \delta_j\}$, each scheme uses as the only inputs
$p$, $n$, and the set of eigenvalues of $\Sample$.
In this paper we derive a new estimator
that is not in the class of \cite{stein:1975,stein:1986}
but applies nonlinear shrinkage to a 
transformation of the data. It thereby implicitly uses more
information than just the sample covariance matrix (together with $p$
and $n$).
Since we focus in
the following on the class of elliptical distributions, we will
differentiate between the dispersion matrix $\mathbf{H}$ and the
covariance matrix $\boldsymbol{\Sigma}$. The former will be defined in Section \ref{sec:method}, but the main difference between the two population quantities is that $\boldsymbol{\Sigma}$ might not
exist. If it does exist, $\boldsymbol{\Sigma}$ is simply
given by $c \mathbf{H}$, with $c > 0$ depending on the
underlying distribution.  

To illustrate the advantage of our method, we now present  a
motivational toy example before moving on to the general
methodology.
We first consider a multivariate Gaussian distribution in dimension $p=200$ with mean
$\boldsymbol{\mu}=\mathbf{0}$ and covariance matrix
$\boldsymbol{\Sigma}=\mathbf{H}$, where the $(i,j)$
element of $\mathbf{H}$ is $0.7^{\mid i - j\mid}$, as
in \cite{linearshrinkageinheavytails}. We simulate $n=300$
i.i.d.\ observations from this distribution. For $j=1,\ldots,p$, the
left panel of Figure \ref{fig:Eig} displays the theoretical optimum
$\delta_j$, $\mathbf{\hat{v}}_{j}^{\top}\Vsol \LQIS
\Vsol\mathbf{\hat{v}}_{j}\eqdef \hat{\delta}_j$, as well as
$\mathbf{\hat{v}}_{j}^{\top}\mathbf{\hat{H}}\mathbf{\hat{v}}_{j}$, where $\mathbf{\hat{H}}$ is the proposed R-NL estimator. Importantly,
the estimated values are very close to the theoretical optimum
$\delta_j$, $j=1,\ldots,p$, for both nonlinear shrinkage and our
proposed method.

We next consider the same setting, but instead simulate from a
multivariate $t$ distribution with $4$ degrees of freedom and
dispersion matrix $\mathbf{H}$, such that the covariance
matrix~$\boldsymbol{\Sigma}$ is $4/(4-2) \cdot \mathbf{H}$. In
particular the
$\mathbf{\hat{v}}_{j}^{\top}\mathbf{\hat{H}}\mathbf{\hat{v}}_{j}$ are
multiplied by $c=2$ in this case to obtain an estimate of
$\mathbf{\hat{v}}_{j}^{\top}\boldsymbol{\Sigma}\mathbf{\hat{v}}_{j}$.
(The value $c=2$ would not be known in practice but is `fair' to use
it in this toy example, since doing so does not favor one estimation
method over the other).
The left panel of Figure \ref{fig:Eig} displays the results.
It can be seen that nonlinear shrinkage overestimates large values of
$\delta_j$ (by a lot) and underestimates small values of $\delta_j$;
on the other hand, our new method does not have this problem and its
performance (almost) matches the one from the Gaussian case.

\begin{figure}
    \vspace{-0.25cm}
    \begin{adjustwidth}{-0cm}{}
    \includegraphics[width=\columnwidth]{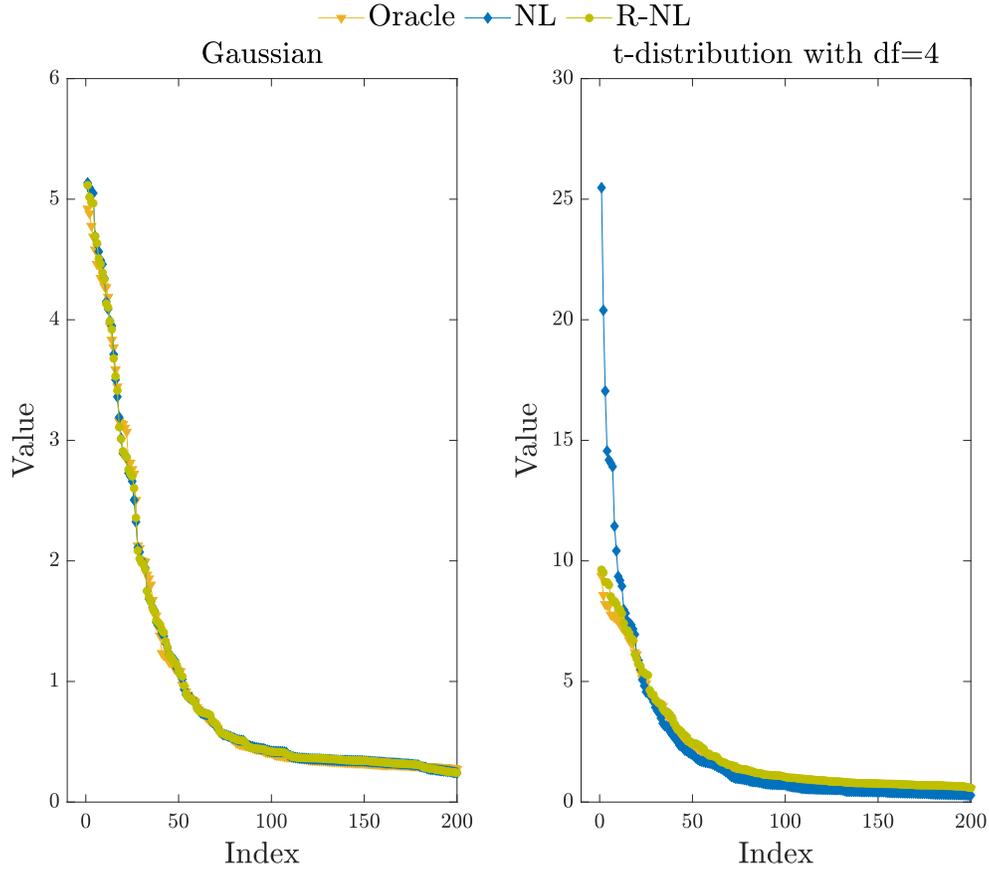}
    \end{adjustwidth}
        \caption{Comparison of the estimated values of nonlinear shrinkage and R-NL. The theoretical optimal eigenvalues
      $\mathbf{\hat{v}}_{j}^{\top}
      \boldsymbol{\Sigma}\mathbf{\hat{v}}_{j}$
      are denoted by
      `Oracle''. In the left panel the sample is taken from a
      multivariate Gaussian distribution and in the right panel from
      a multivariate $t$-distribution with 4~degrees of freedom.
      For both distributions the
      $(i,j)$ element of
      $\mathbf{H}$ is $0.7^{\mid i - j\mid}$, as in
      \cite{linearshrinkageinheavytails}. The number
      of observations is $n=300$ and the dimension is $p=200$.}
     \label{fig:Eig}
\end{figure}


\section{Methodology}
\label{sec:method}

We assume to observe an i.i.d.\ sample $\mathbb{Y}\defeq \{\mathbf{Y}_1, \ldots, \mathbf{Y}_n\} $ from a $p$-dimensional elliptical distribution. If $\mathbf{Y}$ has an elliptical distribution it can be represented as 
\begin{align}\label{distform}
    \mathbf{Y} \stackrel{D}{=} \boldsymbol{\mu} + R \mathbf{H}^{1/2}  \boldsymbol{\xi}~, 
\end{align}
where $R$ is a positive random variable, and $\boldsymbol{\xi}$ is
uniformly distributed on the $p$-dimensional unit sphere, independently
of~$R$, and $\stackrel{D}{=}$ denotes equality in distribution
\citep{ellipticaldisttheory}. The dispersion matrix $\mathbf{H}$ is
assumed to be symmetric positive-definite (pd), with
eigendecomposition $\mathbf{H}=\VT  \LT \VT^{\top} $. If $\mathbf{Y}$
meets \eqref{distform}, we write $\mathbf{Y} \sim
E_{p}(\boldsymbol{\mu}, \mathbf{H}, g )$, where $g$ is the
``generator'' that identifies the distribution of $R$; for example,
see \cite{Kotzelliptical}. We assume this
generator to exist, which is equivalent to $R$ having a density
\citep{Kotzelliptical}. 

In the following we restrict ourselves to distributions of the form
\eqref{distform} with $\boldsymbol{\mu}=\mathbf{0}$ and such that
second moments exist. The assumption $\boldsymbol{\mu}=\mathbf{0}$ is used for simplicity, though it is not necessarily restrictive in the context of elliptical distributions. We refer to the discussion in \cite[Section 2]{thresholdingandTyler}. Then $\Var(\mathbf{Y})=c \mathbf{H}$ for some $c
> 0$. Following \cite{linearshrinkageinheavytails}, we will normalize our estimators of $\mathbf{H}$ to have trace $p$. We note however that, to obtain
an estimator of $\boldsymbol{\Sigma}$, one could instead normalize the
estimator to have the same trace as $\hat{\mathbf{S}}$.  As an illustration,
in the example from the right panel of Figure~\ref{fig:Eig},
$\Tr(\mathbf{\hat{S}}) \approx 428$ whereas 
$\Tr(\boldsymbol{\Sigma}) = 400$.

\subsection{Robust Nonlinear Shrinkage}

We start by outlining our main idea. Let $\|\cdot \|$ be the Euclidean
norm on $\R^p$. As shown in \cite{centralangulargaussian}, if
$\mathbf{Y} \sim E_{p}(\boldsymbol{\mu}, \mathbf{H}, g )$,
$\mathbf{Z}\defeq\frac{\mathbf{Y}}{\|\mathbf{Y} \|}$ has a central
angular Gaussian distribution with density
\begin{align}\label{angularGaussiandist}
    p(\mathbf{z}; \mathbf{H} ) \propto |\mathbf{H}|^{-1/2} \cdot \left(\mathbf{z}^{\top} \mathbf{H}^{-1}\mathbf{z}  \right)^{-p/2}~,
\end{align}
where for $a,b \in \R$, $a \propto b$ means there exists $c > 0$ with $a=cb$. We will also write $\mathbf{A} \propto \mathbf{B}$, for two $p \times p$ matrices $\mathbf{A}, \mathbf{B}$ if $\mathbf{A} =c \mathbf{B}$. The likelihood in \eqref{angularGaussiandist} is the starting point of the original Tyler's method. Taking the derivative of \eqref{angularGaussiandist}, Tyler's estimator $ \HTyler$ is implicitly given by the following condition:
\begin{align}\label{Tylerestimation}
    \HTyler = \frac{p}{n} \sum_{t=1}^n \frac{\mathbf{Z}_t \mathbf{Z}_t^{\top}}{\mathbf{Z}_t^{\top} \HTyler^{-1} \mathbf{Z}_t}~.
\end{align}
This estimator is obtained as the limit of the iterations
\begin{align}\label{Tyleriteration}
    \mathbf{\hat{H}}^{[\ell+1]} \  \propto \  \frac{p}{n} \sum_{t=1}^n \frac{\mathbf{Z}_t \mathbf{Z}_t^{\top}}{\mathbf{Z}_t^{\top} (\mathbf{\hat{H}}^{[\ell]})^{-1} \mathbf{Z}_t}~,
\end{align}
where $\propto$ indicates that $\mathbf{\hat{H}}^{[\ell+1]}$ is
actually obtained after an additional trace-normalization step; for
example, see \cite{Tyler} or~\cite{linearshrinkageinheavytails}. Robust linear shrinkage methods such as the method of~\cite{linearshrinkageinheavytails} augment~\eqref{Tyleriteration} by shrinking towards the identity matrix in each iteration. That is, if for an $p \times p$ matrix $\mathbf{A}$ and $\rho \in [0,1]$, we define $\LS{\mathbf{A}, \rho} \defeq (1-\rho) \mathbf{A} + \rho \mathbf{I} $, then the robust linear shrinkage estimator is obtained from the iterations
\begin{align}\label{LSTyleriteration}
    \mathbf{\hat{H}}^{[\ell+1]} \  \propto \  \LS{ \frac{p}{n} \sum_{t=1}^n \frac{\mathbf{Z}_t \mathbf{Z}_t^{\top}}{\mathbf{Z}_t^{\top} (\mathbf{\hat{H}}^{[\ell]})^{-1} \mathbf{Z}_t}, \rho}~,
\end{align}
where again $\propto$ indicates a trace-normalization step.

Similarly, denote for any symmetric pd  matrix $\mathbf{A}$ by
$\NL{\mathbf{A}}$ the matrix that is obtained when using nonlinear
shrinkage on $\mathbf{A}$. A few clarifications are in order at this
point. First, in the existing literature on nonlinear shrinkage,
$\mathbf{A}$ is always the sample covariance matrix; but the
`algorithm' of nonlinear shrinkage allows for a more general input
instead. Second, there are (at least) three different nonlinear shrinkage
schemes by now: the numerical scheme called QuEST of \cite{Wolf2015},
the analytical scheme of \cite{Analytical_Shrinkage}, and the QIS
method of \cite{QIS2020}, which is also of analytical nature; our
methodology allows for the use of any such scheme, with our personal
choice being the QIS method. Third, any `algorithm' of nonlinear
shrinkage needs as an additional input to  $\mathbf{A}$, which of
course determines the dimension $p$, also the sample size $n$, which
we may treat as fixed and known in our methodology. 

Applying nonlinear shrinkage to the matrix $\mathbf{A}$ leaves its
eigenvectors unchanged and only changes its eigenvalues. The way the
eigenvalues are changed depends on the particular nonlinear shrinkage
scheme; for example, see \cite[Section 4.5]{QIS2020} for the details
concerning the QIS method.
In analogy to the case of linear
shrinkage, we could now apply nonlinear shrinkage each time in the
above iteration. That is, we could iterate
\begin{align}\label{QISTyleriteration}
    \mathbf{\hat{H}}^{[\ell+1]} \  \propto \  \NL{ \frac{p}{n} \sum_{t=1}^n \frac{\mathbf{Z}_t \mathbf{Z}_t^{\top}}{\mathbf{Z}_t^{\top} (\mathbf{\hat{H}}^{[\ell]})^{-1} \mathbf{Z}_t}}~,
\end{align}
where the input to NL corresponds to the sample covariance matrix of the scaled data $\mathbf{Z}_t/(\mathbf{Z}_t^{\top} (\mathbf{\hat{H}}^{[\ell]})^{-1} \mathbf{Z}_t/p)^{1/2}.$
Unfortunately, contrary to the case of linear shrinkage, it is not clear how to ensure convergence for such an approach. However, we note that iteration~\eqref{QISTyleriteration} can be seen as a simultaneous iteration over the eigenvalues and eigenvectors, whereby only the former is changed by nonlinear shrinkage. Following the ideas in \cite{comfortNL}, we instead aim to iterate over the eigenvectors for fixed (shrunken) eigenvalues. That is, after the first iteration, we fix the eigenvalues obtained by nonlinear shrinkage, denoted $\LQIS$. Choosing $\mathbf{\hat{H}}^{[0]}=\mathbf{I}$, this corresponds to using nonlinear shrinkage on the sample covariance matrix of $\mathbb{Z}\defeq \{\mathbf{Z}_1, \ldots, \mathbf{Z}_T \}$, with $\mathbf{Z}_t\defeq \mathbf{Y}_t/\| \mathbf{Y}_t \|$. 
It should be mentioned here that any nonlinear shrinkage scheme ensures that 
the elements on the diagonal of $\LQIS$, denoted
$\hat{\delta}_j$, $j=1,\ldots,p$, are all strictly positive.

We then optimize the likelihood of the central angular Gaussian distribution only with respect to the orthogonal matrix $\VT$. That is, we solve,
\begin{align}\label{problem}
     \Vsol &\defeq  \argmax_{ \mathbf{U} \in \Orth} \sum_{t=1}^{n} \ln(p(\mathbf{Z}_t; \mathbf{U}, \LQIS)) \nonumber \\
     &= \argmin_{ \mathbf{U} \in \Orth}  \frac{1}{n} \sum_{t=1}^{n}  \ln\left( \mathbf{Z}_t^{\top} \mathbf{U} \LQIS^{-1}  \mathbf{U}^{\top} \mathbf{Z}_t  \right)~,
\end{align}
where $\Orth\defeq \{\mathbf{U}: \mathbf{U}^{\top}\mathbf{U}=\mathbf{U}\mathbf{U}^{\top}=\mathbf{I}$\} is the orthogonal group. Finally, once $\Vsol$ is obtained, $\LQIS$ is updated. That is, we apply nonlinear shrinkage to the covariance matrix of the standardized data 
\begin{align}\label{standardized}
     \tilde{\mathbf{Z}}_t\defeq\frac{\mathbf{Z}_t   }{\sqrt{\mathbf{Z}_t^{\top}  \Vsol \LQIS^{-1}  \Vsol^{\top}  \mathbf{Z}_t/p}}~, t=1,\ldots,n,
\end{align}
to obtain $\LRQIS$. The final estimate is then given as 
\begin{align}\label{Hestimate}
    \mathbf{\hat{H}}\defeq p\cdot \Vsol \LRQIS \Vsol^{\top}/\Tr( \Vsol \LRQIS \Vsol^{\top})~. 
\end{align}
Since, as we will show below, the eigenvectors of the sample covariance matrix of $\{\mathbf{\tilde{Z}}_1, \ldots, \mathbf{\tilde{Z}}_n\}$ are again given by $\Vsol$, it holds that,
\begin{align*}
   \mathbf{\hat{H}} \  \propto \  \NL{ \frac{1}{n} \sum_{t=1}^{n}\frac{\mathbf{Z}_t \mathbf{Z}_t^{\top}  }{\mathbf{Z}_t^{\top}  \Vsol \LQIS^{-1}  \Vsol^{\top}  \mathbf{Z}_t}/p}~.
\end{align*}
The whole procedure is summarized in Algorithm \ref{RNL}. We now detail how to solve \eqref{problem}.

For an $p \times p$ symmetric pd matrix $\mathbf{A}$, let 
\[
\mathbf{A}=\mathbf{U}_A \boldsymbol{\Lambda }_A\mathbf{U}_A^{\top}\ ,
\]
be its eigendecomposition, where we assume the elements of $\boldsymbol{\Lambda }_A$ to be ordered from smallest to largest. We define $\mathcal{E}$ to be the operator that returns all possible matrices of eigenvectors. That is, $\mathcal{E}(\mathbf{A})$ is a subset of $\Orth$ and for any $\mathbf{U} \in \mathcal{E}(\mathbf{A})$, $\mathbf{U}^{\top} \mathbf{A} \mathbf{U}$ is a diagonal matrix with elements ordered from smallest to largest.

We also define in the following for $\mathbf{U} \in \Orth$,
\begin{align}\label{fdef}
    \f{\mathbf{U}}\defeq\sum_{t=1}^{n}\frac{\mathbf{Z}_t \mathbf{Z}_t^{\top} }{ \mathbf{Z}_t^{\top} \mathbf{U} \LQIS^{-1}  \mathbf{U}^{\top} \mathbf{Z}_t }~,
\end{align}
where the dependence on $\mathbf{Z}_1,\dots,\mathbf{Z}_n$ and $\LQIS$ is suppressed to keep notation compact.

\begin{restatable}{lemma}{lone}\label{existence}
A minimizer $\Vsol$ of \eqref{problem} exists and meets the condition
\begin{align}\label{critical_conditions}
  \Vsol^{\top} \f{\Vsol} \Vsol \LQIS^{-1}   = \LQIS^{-1}  \Vsol^{\top} \f{\Vsol} \Vsol~.
\end{align}
\end{restatable}

\begin{proof}


Since the orthogonal group is compact \citep[Ch. 3]{MahonyBook} and 
\begin{align}
    \mathbf{U} \mapsto f(\mathbf{U}):= \frac{1}{n}\sum_{t=1}^{n}  \ln\left( \mathbf{Z}_t^{\top} \mathbf{U} \LQIS^{-1}  \mathbf{U}^{\top} \mathbf{Z}_t  \right)~,
\end{align}
is continuous, $f(\mathbf{U})$ takes its minimal and maximal value on $\Orth$. Thus there exists $\Vsol \in \Orth$ such that $\Vsol$ minimizes $f$.

On the other hand, according to \cite{Wen2013}, if $\Vsol$ is a minimizer of $f$, it must satisfy the following first-order conditions: 
\begin{align*}
    &G\Vsol^{\top} - \Vsol G^{\top} = \mathbf{0}~,
\end{align*}
where $G$ is the unconstrained gradient of problem \eqref{problem},
\begin{align*}
    G:=\frac{1}{n}\sum_{t=1}^{n}  \frac{1}{ \mathbf{Z}_t^{\top}  \Vsol \LQIS^{-1}  \Vsol^{\top} \mathbf{Z}_t }  \mathbf{Z}_t \mathbf{Z}_t^{\top} \Vsol\LQIS^{-1}~.
\end{align*}
Thus
\begin{align*}
        &\sum_{t=1}^{n} \left( \frac{\mathbf{Z}_t \mathbf{Z}_t^{\top} \Vsol\LQIS^{-1} \Vsol^{\top}}{ \mathbf{Z}_t^{\top} \Vsol \LQIS^{-1}  \Vsol^{\top} \mathbf{Z}_t }  -  \frac{\Vsol  \LQIS^{-1} \Vsol^{\top}  \mathbf{Z}_t \mathbf{Z}_t^{\top}  }{ \mathbf{Z}_t^{\top} \Vsol \LQIS^{-1}  \Vsol^{\top} \mathbf{Z}_t } \right)   = \mathbf{0}\\
    &  \Vsol^{\top} \sum_{t=1}^{n}\frac{\mathbf{Z}_t \mathbf{Z}_t^{\top} }{ \mathbf{Z}_t^{\top} \Vsol \LQIS^{-1}  \Vsol  \LQIS^{-1}  \mathbf{Z}_t }\Vsol \LQIS^{-1} \\ &- \LQIS^{-1}  \Vsol^{\top} \sum_{t=1}^{n}\frac{\mathbf{Z}_t \mathbf{Z}_t^{\top} }{ \mathbf{Z}_t^{\top} \Vsol \LQIS^{-1}  \Vsol^{\top} \mathbf{Z}_t } \Vsol  = \mathbf{0}~.
\end{align*}
Hence, any minimizer $\Vsol \in \Orth$ meets \eqref{critical_conditions}.\\
\null \qed
\end{proof}


The necessary condition in \eqref{critical_conditions} is true in particular if $\Vsol$ diagonalizes $\f{\Vsol}$, or
\begin{align}\label{recursion}
      \Vsol \in  \mathcal{E} \left( \frac{1}{n} \sum_{t=1}^n \frac{\mathbf{Z}_t \mathbf{Z}_t^{\top} }{\mathbf{Z}_t^{\top} \Vsol \LQIS^{-1}  \Vsol^{\top} \mathbf{Z}_t} \right)~,
\end{align}
in analogy to \eqref{Tylerestimation}. Thus given $\LQIS$, we propose the following iterations
\begin{align}
    \Vit{\ell + 1}& \in \mathcal{E} \left( \frac{1}{n} \sum_{t=1}^n \frac{\mathbf{Z}_t \mathbf{Z}_t^{\top}  }{\mathbf{Z}_t^{\top}  \Vit{\ell} \LQIS^{-1}  (\Vit{\ell})^{\top}  \mathbf{Z}_t} \right)~, \label{iteration1}
\end{align}
starting with
\begin{align}
    \Vit{1} \in \mathcal{E} \left( \frac{1}{n} \sum_{t=1}^n \frac{\mathbf{Z}_t \mathbf{Z}_t^{\top}  }{\| \mathbf{Z}_t \|} \right)=\mathcal{E} \left( \frac{1}{n} \sum_{t=1}^n \mathbf{Z}_t \mathbf{Z}_t^{\top}   \right)~.
\end{align}

As noted above, this corresponds to the iterations in \cite[Algorithm
1]{ourclosestcompetitor}, for $\alpha=\infty$. We also note that the
time complexity in each iteration is the same as for the iterations of
the original Tyler's method in \eqref{Tyleriteration}, namely $O(np^2
+ p^3)$; for example, see \cite{timecomplexityTylerest}.

We now proceed by showing that any sequence generated by these iterations has a limit $\Vit{\infty}$ such that \eqref{critical_conditions} holds. To this end we adapt the approach taken in \cite{Wieselhottheory, majorizationpaper, existence_uniqueness_algorithms} and define the surrogate function
\begin{align}\label{surrogate}
   g(\mathbf{U}\mid \Vit{\ell}  ) \defeq &\frac{1}{n}\sum_{t=1}^{n}  \ln\left( \mathbf{Z}_t^{\top} \Vit{\ell} \LQIS^{-1}  (\Vit{\ell})^{\top} \mathbf{Z}_t  \right) \nonumber \\ &+\frac{1}{n}\sum_{t=1}^{n} \frac{\mathbf{Z}_t^{\top} \mathbf{U} \LQIS^{-1}  \mathbf{U}^{\top} \mathbf{Z}_t}{\mathbf{Z}_t^{\top} \Vit{\ell} \LQIS^{-1}  (\Vit{\ell})^{\top} \mathbf{Z}_t} - 1~.
\end{align}
Then for $f(\mathbf{U}) \defeq \frac{1}{n}\sum_{t=1}^{n}  \ln\left( \mathbf{Z}_t^{\top} \mathbf{U} \LQIS^{-1}  \mathbf{U}^{\top} \mathbf{Z}_t  \right)$,

\begin{restatable}{lemma}{lthree}\label{surrogateproperites}
The surrogate function $g$ satisfies:
\begin{align}
    f(\mathbf{U}) &\leq g(\mathbf{U}\mid \Vit{\ell}  ) \text{ for all } \mathbf{U}, \Vit{\ell} \in \Orth\\
    f(\Vit{\ell}) &= g(\Vit{\ell}\mid \Vit{\ell}  )~,
\end{align}
and for $\Vit{\ell + 1}$ as in \eqref{iteration1},
\begin{align}
    g(\Vit{\ell + 1}\mid \Vit{\ell}  ) \leq g(\mathbf{U}\mid \Vit{\ell}  ) \text{ for all } \mathbf{U} \in \Orth~.
\end{align}
\end{restatable}

\begin{proof}
  The first inequality follows from the fact that, for~any \mbox{$a > 0$},
  $\log(x) \leq \log(a)
+ (\frac{x}{a} - 1)$ (\cite{Wieselhottheory}), whereas the second
equality is trivial. For the last claim, we can write
\begin{align*}
    &\argmin_{\mathbf{U} \in \Orth} g(\mathbf{U}\mid \Vit{\ell}  )\\
    &=  \argmin_{\mathbf{U} \in \Orth} \mbox{Tr}\left(\frac{1}{n}\sum_{t=1}^{n} \frac{\mathbf{Z}_t \mathbf{Z}_t^{\top}}{\mathbf{Z}_t^{\top} \Vit{\ell} \LQIS^{-1}  (\Vit{\ell})^{\top} \mathbf{Z}_t}  \mathbf{U} \LQIS^{-1}  \mathbf{U}^{\top}\right)~.
\end{align*}
Since we assume $\LQIS$ has ordered values, this is globally minimized
when $\mathbf{U}$ is chosen to diagonalize $\frac{1}{n}\sum_{t=1}^{n}
\frac{\mathbf{Z}_t \mathbf{Z}_t^{\top}}{\mathbf{Z}_t^{\top} \Vit{\ell}
  \LQIS^{-1}  (\Vit{\ell})^{\top} \mathbf{Z}_t}$; for example, see \cite{comfortNL}.\\
\null \qed
\end{proof}

Define now the set of critical points as $\Vset \subset \Orth$, that is,
\[
\Vset:=\{\Vsol \in \Orth \text{ such that } \eqref{critical_conditions} \text{ holds}\}~,
\]
and let for all $\mathbf{U} \in \Orth$,
\begin{align*}
    d(\mathbf{U},\Vset) := \inf_{\Vsol \in \Vset} \|\Vsol-\mathbf{U} \|_{F}~,
\end{align*}
as in \cite{majorizationpaper}. Using Lemma \ref{surrogateproperites} the following convergence result can be obtained.

\begin{restatable}{theorem}{thmone}\label{awesometheorem}
For any sequence $(\Vit{\ell})_{\ell=1}^{\infty}$ generated by the above iterations,
\begin{align}\label{monotonicity}
    f(\Vit{\ell + 1}) \leq f(\Vit{\ell})~, \text{ for all } \ell,
\end{align}
and 
\begin{align}\label{convergence}
 \lim_{\ell \to \infty }d(\Vit{\ell},\Vset) = 0~.
\end{align}
\end{restatable}

\begin{proof}
The proof closely follows the argument in \cite[Theorem 1, Corollary 1]{majorizationpaper}. Using Lemma \ref{surrogateproperites}, we have that for all $\ell$,
\begin{align*}
    f(\Vit{\ell + 1}) \leq g(\Vit{\ell + 1} \mid \Vit{\ell}) \leq  g(\Vit{\ell} \mid \Vit{\ell}) =f(\Vit{\ell})~,
\end{align*}
proving the first part. For the second, since the orthogonal group $\Orth$ is compact, there exists a subsequence $(\Vit{\ell_k})_{k=1}^{\infty}$ of  $(\Vit{\ell})_{\ell=1}^{\infty}$ that converges to $\Vit{\infty} \in \Orth$. Additionally, for all $\mathbf{U} \in \Orth$,
\begin{align*}
    &g(\mathbf{U} \mid \Vit{\ell_k}) \geq g(\Vit{\ell_{k} + 1 } \mid \Vit{\ell_k}) \geq f(\Vit{\ell_{k} + 1 })\\
    &\geq f(\Vit{\ell_{k+1}}) =g(\Vit{\ell_{k+1}} \mid \Vit{\ell_{k+1}})~,
\end{align*}
since by the properties of subsequences $\ell_{k+1} \geq \ell_k +1 $. Letting $k \to \infty$, thanks to the joint continuity of $(\mathbf{U}_1, \mathbf{U}_2) \mapsto g(\mathbf{U}_1 \mid \mathbf{U}_2)$, this implies
\begin{align*}
   g(\mathbf{U} \mid \Vit{\infty}) \geq   g(\Vit{\infty} \mid \Vit{\infty})~, \text{for all } \mathbf{U} \in \Orth .
\end{align*}
Thus $\Vit{\infty}$ is the global minimizer of the function $\mathbf{U} \mapsto g(\mathbf{U} \mid \Vit{\infty})$. In particular, the first-order conditions must hold: Thus
\begin{align*}
     &G(\Vit{\infty})^{\top} - \Vit{\infty} G^{\top} = \mathbf{0}~,
\end{align*}
where $G$ is the unconstrained derivative at $\Vit{\infty}$:
\begin{align*}
     G:=\frac{1}{n}\sum_{t=1}^{n}  \frac{1}{ \mathbf{Z}_t^{\top}  \Vit{\infty} \LQIS^{-1}  (\Vit{\infty})^{\top} \mathbf{Z}_t }  \mathbf{Z}_t \mathbf{Z}_t^{\top} \Vit{\infty}\LQIS^{-1}~.
\end{align*}
Thus it holds that
\begin{align*}
&(\Vit{\infty})^{\top} \sum_{t=1}^{n}\frac{\mathbf{Z}_t \mathbf{Z}_t^{\top} }{ \mathbf{Z}_t^{\top} (\Vit{\infty})^{\top} \LQIS^{-1}  \Vit{\infty}  \LQIS^{-1}  \mathbf{Z}_t }\Vit{\infty} \LQIS^{-1}  \\ &- \LQIS^{-1}  (\Vit{\infty})^{\top} \sum_{t=1}^{n}\frac{\mathbf{Z}_t \mathbf{Z}_t^{\top} }{ \mathbf{Z}_t^{\top} \Vit{\infty} \LQIS^{-1}  (\Vit{\infty})^{\top} \mathbf{Z}_t } \Vit{\infty}  = \mathbf{0}~,
\end{align*}
which corresponds to the desired first-order conditions for the minimization of $f$ and thus $\Vit{\infty}  \in \Vset$, or $d(\Vit{\infty}, \Vset)=0$.

Repeating this argument, it follows that any subsequence of $(\Vit{\ell})_{\ell=1}^{\infty}$ has a further subsequence converging to some $\Vit{\infty}$ (depending on the subsequence) with $d(\Vit{\infty}, \Vset)=0$. Now assume the overall sequence does not converge to a point in $\Vset$. Then there is a subsequence $(\Vit{\ell_k})_{k=1}^{\infty}$ such that for all $k$
\begin{align*}
    d(\Vit{\ell_k}, \Vset) \geq \varepsilon ~,
\end{align*}
for some $\varepsilon>0$. But then this would be true also for any subsequence, a contradiction.\\
\null \qed
\end{proof}

Thus, as $\ell \to \infty$, $\Vit{\ell}$ gets arbitrary close to a critical point. This leads to the following convergence criterion:
\begin{align*}
    &\|(\Vit{\ell-1})^{\top} \f{\Vit{\ell-1}} \Vit{\ell-1} \LQIS^{-1} \\ &- \LQIS^{-1}  (\Vit{\ell})^{\top} \f{\Vit{\ell}} \Vit{\ell}\|_F\leq \epsilon~,
\end{align*}
where $\epsilon>0$ is some convergence tolerance. This criterion is
used in Algorithm \ref{VIteration} and we set $\epsilon=10^{-10}$ in Sections \ref{sec:mc} and~\ref{sec:empirics}.

Although \mbox{R-NL} is no longer in the same class of estimators as nonlinear
shrinkage, namely the class of \cite{stein:1975,stein:1986},
an interesting question is whether it is still \emph{rotation-equivariant}. An estimator $\mathbf{\hat{H}}$ applied to $\mathbb{Y}=\{\mathbf{Y}_1, \ldots, \mathbf{Y}_T\}$ is rotation-equivariant if, for any rotation $\mathbf{R}$ and rotated data $\mathbf{\bar{Y}}_t\defeq \mathbf{R} \mathbf{Y}_t$, $t=1,\ldots,n$, the estimate of the rotated data, $\mathbf{\hat{H}}_{\mathbf{R}}$, satisfies
\begin{align}\label{H_R}
    \mathbf{\hat{H}}_{\mathbf{R}}= \mathbf{R} \mathbf{\hat{H}} \mathbf{R}^{\top}~.
\end{align}
This is true for any estimator in the class of
\cite{stein:1975,stein:1986} and, therefore, in particular for
nonlinear shrinkage.
We now show that this is true for \mbox{R-NL} as well, using the following lemma:

\begin{restatable}{lemma}{lfour}\label{rotationeq}
Let $\mathbf{R}$ be an arbitrary rotation, $\mathbf{\bar{Z}}_t\defeq \mathbf{R} \mathbf{Z}_t$ and $\Vit{\ell}_r$ be the $\ell$th iteration of Algorithm \ref{VIteration} applied to $\{\mathbf{\bar{Z}}_1, \ldots, \mathbf{\bar{Z}}_T \}$. Then
\begin{align}
   &\mathcal{E} \left( \frac{1}{n} \sum_{t=1}^n \frac{\mathbf{\bar{Z}}_t \mathbf{\bar{Z}}_t^{\top}  }{\mathbf{\bar{Z}}_t^{\top}  \Vit{\ell}_r \LQIS^{-1}  (\Vit{\ell}_r)^{\top}  \mathbf{\bar{Z}}_t} \right)=\left\{\mathbf{R} \Vit{\ell +1 }:\right. \nonumber \\
   &\left. \Vit{\ell + 1} \in \mathcal{E} \left( \frac{1}{n} \sum_{t=1}^n \frac{\mathbf{Z}_t \mathbf{Z}_t^{\top}  }{\mathbf{Z}_t^{\top}  \Vit{\ell} \LQIS^{-1}  (\Vit{\ell})^{\top}  \mathbf{Z}_t} \right) \right\}~.
    \label{rotationsolution}
\end{align}
\end{restatable}

\begin{proof}

It clearly holds that \eqref{rotationsolution} is true for $\ell=0$. Assume \eqref{rotationsolution} holds for $\ell$, we show that it holds for $\ell + 1$: By assumption we can write $\Vit{\ell}_r=\mathbf{R} \Vit{\ell}$. Thus 
\begin{align*}
  &\mathcal{E} \left( \frac{1}{n} \sum_{t=1}^n \frac{\mathbf{\bar{Z}}_t \mathbf{\bar{Z}}_t^{\top}  }{\mathbf{\bar{Z}}_t^{\top}  \Vit{\ell}_r \LQIS^{-1}  (\Vit{\ell}_r)^{\top}  \mathbf{\bar{Z}}_t} \right)\\
  &=\mathcal{E} \left(\mathbf{R} \frac{1}{n} \sum_{t=1}^n \frac{\mathbf{Z}_t \mathbf{Z}_t^{\top}  }{\mathbf{Z}_t^{\top}  \Vit{\ell} \LQIS^{-1}  ( \Vit{\ell})^{\top}  \mathbf{Z}_t} \mathbf{R}^{\top} \right)\\
  &=\mathbf{R}\mathcal{E} \left( \frac{1}{n} \sum_{t=1}^n \frac{\mathbf{Z}_t \mathbf{Z}_t^{\top}  }{\mathbf{Z}_t^{\top}  \Vit{\ell} \LQIS^{-1}  ( \Vit{\ell})^{\top}  \mathbf{Z}_t}  \right)~,
\end{align*}
 and thus \eqref{rotationsolution} hold true.\\
\null \qed
\end{proof}

We note that the rotation of the original data in $\mathbb{Y}$ corresponds to a rotation of $\mathbb{Z}$, since for all $t$, $\mathbf{R} \mathbf{Y}_t/\|\mathbf{R} \mathbf{Y}_t \|=\mathbf{R} \mathbf{Y}_t/\|\mathbf{Y}_t \|=\mathbf{R} \mathbf{Z}_t$. Thus if $(\Vit{\ell})_{\ell=1}^{\infty}$ is a sequence generated by Algorithm \ref{VIteration} for the data $\{\mathbf{Y}_1, \ldots, \mathbf{Y}_T \}$, then $(\mathbf{R}\Vit{\ell})_{\ell=1}^{\infty}$ is a sequence generated for the rotated data. Consequently, the corresponding estimate of $\mathbf{H}_{\mathbf{R}}$ for the rotated data will be of the form \eqref{H_R}.

\subsection{Uniqueness}\label{practiceimprovsec}


The matrix $\Vsol$ in \eqref{recursion} (and consequently in
\eqref{critical_conditions}) is not unique in general. However, we are
ultimately not interested in $\Vsol$, but in $\Vsol \LQIS
\Vsol^{\top}$. A natural question is thus whether $\Vsol \LQIS
\Vsol^{\top}$ is unique even if $\Vsol$ is not. This turns out to be true with probability one, as we detail now. Let in the
following $\EigFVhat$ be the diagonal matrix of (ordered) eigenvalues
of $F \bigl (\Vsol \bigr )$. It has the form $\EigFVhat=\EigFVhat_{+}$, if $p < n$ and 
\begin{align*}
    \EigFVhat= \begin{pmatrix} \mathbf{0}_{p-n+1, p-n+1} & \mathbf{0}_{p-n+1, n-1}\\
\mathbf{0}_{n-1, p-n+1} & \EigFVhat_{+}\end{pmatrix} , & \text{ if } p \geq n~,
\end{align*}
where $\EigFVhat_{+}$ is a diagonal matrix with the largest $\min(n-1,p)$ eigenvalues and $\mathbf{0}_{n, p}$ is an $n \times p$ matrix of zeros. We denote the diagonal elements of $\EigFVhat$ as $\hat{\lambda}_1, \ldots \hat{\lambda}_p$. Define similarly the matrix of eigenvalues of $\f{\Vit{\ell-1}}$ as $\EigFl{\ell}$, with elements $\hat{\lambda}_1^{[\ell]}, \ldots \hat{\lambda}_p^{[\ell]}$, and $\EigFVhat_{+}^{[\ell]}$ as the largest $\min(n-1,p)$ eigenvalues of $\EigFl{\ell}$.

\begin{restatable}{lemma}{ltwo}\label{unique}
Assume that whenever $\hat{\lambda}_i=\hat{\lambda}_j$, also $\hat{\delta}_i=\hat{\delta}_j$, for $j,i \in \{1,\ldots,p \}$. Then if $\Vsol_1$ and $\Vsol_2$ meet \eqref{recursion},
\[
\Vsol_1\LQIS (\Vsol_1)^{\top} =\Vsol_2 \LQIS (\Vsol_2)^{\top}.
\]
\end{restatable}

\begin{proof}
If $\hat{\lambda}_i$ is unique, the corresponding eigenvector $\mathbf{v}_i$ is the basis of the one-dimensional space $\{\mathbf{u}: (F(\Vsol) - \hat{\lambda}_i \mathbf{I})\mathbf{u}=\mathbf{0}  \}$. As such if $\mathbf{v}_i^1$, $\mathbf{v}_i^2$ are the $i$th column of $\Vsol_1 $ and $\Vsol_2$ respectively, it must hold that $\mathbf{v}_i^1=\mathbf{v}_i^2$ or $\mathbf{v}_i^1=-\mathbf{v}_i^2$. However as $\Vsol_1 \LQIS \Vsol_1^{\top}= \sum_{i=1}^{p} \hat{\delta}_i \mathbf{v}_i^1 (\mathbf{v}_i^1)^{\top}$ this does not affect the overall matrix. This holds true whether or not $\hat{\delta}_i$ in $\LQIS$ is unique.

Now assume there is $\hat{\lambda}_i$ with multiplicity $p_0$, whereas all other $\hat{\lambda}_j$ are unique. By assumption, $\LQIS$ mimics this pattern and we can reorder their values such that: 
\[
\EigFVhat=\begin{pmatrix} \EigFVhat_1 & \mathbf{0}_{p-p_0, p_0}\\
\mathbf{0}_{p_0, p- p_0} & \EigFVhat_2\end{pmatrix} \text{ and }
\LQIS=\begin{pmatrix} \LQIS^1 & \mathbf{0}_{p-p_0, p_0}\\
\mathbf{0}_{p_0, p- p_0} & \LQIS^2\end{pmatrix}~,
\]
where $\EigFVhat_1$ contains unique ordered values and $\EigFVhat_2$ of size \mbox{$p_0 \times p_0$} contains one value with multiplicity. By assumption, $\LQIS^1$ might have values with multiplicity larger one, but $\LQIS^2$ also contains only copies of one value. We similarly decompose the newly ordered $\Vsol_{1}, \Vsol_{2}$:
\[
\Vsol_{1}=[\Vsol_{11}, \Vsol_{12}] \text{ and } \Vsol_{2}=[\Vsol_{21}, \Vsol_{22}]~.
\]
The columns of $ \Vsol_{12},  \Vsol_{22}$ now form an orthogonal basis of the $p_0$-dimensional eigenvectorspace. As such, we can express each column of $\Vsol_{12}$ as a linear combination of columns in $ \Vsol_{22}$, that is there exists $\mathbf{A} \in \R^{p_0 \times p_0}$, such that $\Vsol_{12}=\Vsol_{22} \mathbf{A}$. Moreover
\begin{align*}
    \mathbf{I}=\Vsol_{12}^{\top} \Vsol_{12} = \mathbf{A}^{\top} \Vsol_{22}^\top   \Vsol_{22} \mathbf{A}=\mathbf{A}^{\top}\mathbf{A}~.
\end{align*}
Thus the columns of $\mathbf{A}$ are orthogonal and since it is square, it has full rank and $\mathbf{A}\mathbf{A}^{\top}=\mathbf{I}$ holds as well. Finally,
\begin{align*}
\Vsol_1 \LQIS \Vsol_1^{\top} &= \Vsol_{11} \LQIS^1 \Vsol_{11}
^{\top} + \Vsol_{12} \LQIS^2 \Vsol_{12}^{\top} \\
&=     \Vsol_{21} \LQIS^1 \Vsol_{21}^{\top} + \hat{\delta}_i \Vsol_{12} \Vsol_{12}^{\top}\\
&=     \Vsol_{21} \LQIS^1 \Vsol_{21}^{\top} + \hat{\delta}_i \Vsol_{22} \mathbf{A} \mathbf{A}^{\top} \Vsol_{22}^{\top}\\
&=     \Vsol_2 \LQIS \Vsol_2^{\top}~.
\end{align*}
A similar approach can be used to show that the equality holds if several $\hat{\lambda}_i$ have multiplicity larger than one.\\
\null \qed
\end{proof}

Thus if the multiplicity of eigenvalues of $\EigFVhat$ implies the multiplicity of the corresponding eigenvalue in $\LQIS$, the resulting matrix will also be the same. This is true in particular if $\hat{\lambda}_i \neq \hat{\lambda}_j$ for all $i,j$. Consequently, under the conditions of Lemma \ref{unique}, $\Vsol$ in \eqref{recursion} is unique under the equivalence relation $\eqiv$  with
\[
\mathbf{U}_1 \eqiv \mathbf{U}_2 \iff \mathbf{U}_1 \LQIS (\mathbf{U}_1)^{\top} =\mathbf{U}_2 \LQIS (\mathbf{U}_2)^{\top}.
\]

More generally if we consider the space of equivalence classes $\Ortheq:=\Orth\setminus \eqiv $ and define the metric 
\[
\tilde{d}_{\LQIS}([\mathbf{U}_1], [\mathbf{U}_2]):=\| \mathbf{U}_1 \LQIS \mathbf{U}_1^{\top} -  \mathbf{U}_2 \LQIS \mathbf{U}_2^{\top} \|_F,
\]
where $[\mathbf{U}]:=\{\mathbf{U}_0 \in \Orth: \mathbf{U}_0\eqiv
\mathbf{U}\} $, we obtain the following lemma.

\begin{restatable}{lemma}{lfive}\label{uniqueoverell}
Assume that 
\begin{align}\label{cond}
\forall \ell \ \ \hat{\lambda}_i^{[\ell]}=\hat{\lambda}_j^{[\ell]} \implies \hat{\delta}_i=\hat{\delta}_j.
\end{align}
Then we can write iteration \eqref{iteration1} in terms of equivalence classes:
\begin{align}\label{iterationwithequiv}
    [\Vit{\ell + 1}] = \mathcal{E}( F([\Vit{\ell}]))~.
\end{align}
Moreover there exists $[\Vsol] \in \Ortheq$ such that \eqref{recursion} holds and the generated sequence $\left([\Vit{\ell}]\right)_{\ell=1}^{\infty}$ satisfies
\begin{align*}
    \tilde{d}_{\LQIS}([\Vit{\ell}], [\Vsol]) \to 0~.
\end{align*}
\end{restatable}

\begin{proof}
First we note that, since the values in $\LQIS$ are all strictly larger than zero,
\[
\mathbf{U}_1 \eqiv \mathbf{U}_2 \iff \mathbf{U}_1 \LQIS^{-1} (\mathbf{U}_1)^{\top} =\mathbf{U}_2 \LQIS^{-1} (\mathbf{U}_2)^{\top}.
\]
Thus for any two $\mathbf{U}_1 \in [\mathbf{U}]$, $\mathbf{U}_2 \in [\mathbf{U}]$, $F(\mathbf{U}_1)=F(\mathbf{U}_2)$, such that we may write $F$ directly as a function of the equivalence class, $F([\mathbf{U}])$. Moreover, since the eigenvalues of $F([\mathbf{V}]^{[\ell]})$ meet the multiplicity condition, the same proof as in Lemma \ref{unique} gives that any $\mathbf{U}_1 \in \mathcal{E}(F(\Vit{\ell}))$, $\mathbf{U}_2 \in \mathcal{E}(F(\Vit{\ell}))$ have $\mathbf{U}_1 \eqiv \mathbf{U}_2 $. Thus \eqref{iterationwithequiv} holds and we can write $[\mathbf{V}]^{[\ell]}=[\Vit{\ell}]$.

Moreover, by the same argument as above, the function value $f(\mathbf{U})$ of any member of an equivalence class is the same, such that we may again write $f([\mathbf{U}])$. Finally $\Ortheq$ is still compact with the metric $\tilde{d}_{\LQIS}$. Indeed consider a sequence $([\mathbf{U}_n])_n$ in $\Ortheq$. For each $n$ we choose an arbtriary representative $\mathbf{U}_n \in \Orth$, to form the sequence $\left( \mathbf{U}_n \right)_n$. Since $\Orth$ is compact, this sequence will have a convergent subsequence $\left( \mathbf{U}_{n_k} \right)_k$. We now show that the corresponding subsequence in $\Orth$, $\left( [\mathbf{U}]_{n_k} \right)_k$ converges in $\Ortheq$. Indeed notice that for any convergent sequence, that is, $\left( \mathbf{U}_n \right)_n$ such that $\mathbf{U}_n \to \mathbf{U}$, it follows by the continuity of the matrix product that
\begin{align*}
    \mathbf{U}_n \LQIS  \mathbf{U}_n ^{\top} \to \mathbf{U} \LQIS  \mathbf{U}^{\top}~,
\end{align*}
or 
\begin{align*}
    \tilde{d}_{\LQIS}([\mathbf{U}_n], [\mathbf{U}]) = \| \mathbf{U}_n \LQIS  \mathbf{U}_n ^{\top} \to \mathbf{U} \LQIS  \mathbf{U}^{\top} \|_F \to 0~.
\end{align*}
Applying this to $\left( [\mathbf{U}]_{n_k} \right)_k$, $ \tilde{d}_{\LQIS}([\mathbf{U}_{n_k}], [\mathbf{U}]) \to 0$. Since the sequence was arbitrary, every sequence in $\Ortheq$ has a convergent subsequence in $\Ortheq$ and thus $(\Ortheq, \tilde{d}_{\LQIS})$ is compact.

Finally we can trace the same steps as in Theorem \ref{awesometheorem} to show that
\begin{align*}
     \inf_{[\Vsol] \in \Vset_0} \tilde{d}_{\LQIS}([\Vit{\ell}],[\Vsol] ) \to 0~,
\end{align*}
where now the set $\Vset_0 \subset \Ortheq$ such that \eqref{recursion} holds has only one member $[\Vsol]$.\\
\null \qed
\end{proof}

At first, it might seem unclear how to enforce the eigenvalue condition in Lemma \ref{uniqueoverell}. However, since $\hat{\lambda}_i^{[\ell]}$, $i=1,\ldots,p$ are eigenvalues of the sample covariance matrix of the standardized sample, Theorem 1 of \cite{Amazingsampleeigenvalueresult} applies. This implies that the eigenvalues of $\EigFVhat_{+}^{[\ell]}$ are all nonzero and distinct with probability one. Thus we only need to ensure that, for $p \geq n$, the smallest $p-n+1$ eigenvalues in $\LQIS$ are all the same. This is enforced in Algorithm
\ref{VIteration} by simply setting the smallest $p-n+1$ values of $\LQIS$ to the value with the highest
multiplicity. The following lemma now obtains.

\begin{restatable}{lemma}{none}
Condition \eqref{cond} holds with probability one.
\end{restatable}

We thus obtain uniqueness of $\Vsol \LQIS \Vsol^{\top}$ and of $\mathbf{\hat{H}}$, up to scaling.

\begin{algorithm}
\normalsize
\textbf{Inputs}: centered data $\mathbb{Z}$, eigenvalue matrix $\LQIS$\;
\textbf{Output}: $\Vsol$\;
\textbf{Hyper-parameters}: Convergence Tolerance $\epsilon$ \;

Initiate $\ell=0$, $c(0)=-\infty$, $c(1)=0$, $\Vit{0}=\mathbf{I}$\;
 \While{$c(\ell+1)-c(\ell) > \epsilon  $}{
  - Calculate $\f{\Vit{\ell}}$ as in \eqref{fdef} and its eigendecomposition $\mathbf{U} \EigFl{\ell+1} \mathbf{U}^{\top}$, with $\EigFl{\ell+1}$ ordered\;
  - Take $\Vit{\ell + 1}=\mathbf{U}$\;
  - $\ell=\ell+1$\;
- $c(\ell+1)= \|(\Vit{\ell-1})^{\top} \f{\Vit{\ell-1}} \Vit{\ell-1} \LQIS^{-1} - \LQIS^{-1}  (\Vit{\ell})^{\top} \f{\Vit{\ell}} \Vit{\ell}\|_F$\;
 }
\Return{$\Vsol=\Vit{\ell}$}
 \caption{VIteration($\mathbb{Z}$, $\LQIS$)}
 \label{VIteration}
\end{algorithm}

\begin{algorithm}
\normalsize
\textbf{Inputs}: centered data $\mathbb{Y}$\;
\textbf{Output}: $\mathbf{\hat{H}}$\;
- Calculate $\mathbb{Z}\defeq \{\mathbf{Z}_1, \ldots \mathbf{Z}_T \}$ with $\mathbf{Z}_t\defeq \mathbf{Y}_t/\| \mathbf{Y}_t\|$, $t=1,\ldots,n$\;
- Obtain the sorted eigenvalues $\LQIS$ by applying NL to $\mathbb{Z}$\;
- If $p \geq n$: Ensure that the last $p-n+1$ elements of $\LQIS$ are equal\;
- Obtain $\Vsol\defeq $ VIteration($\mathbb{Z}, \LQIS$) using Algorithm \ref{VIteration}\;
- Calculate $\tilde{\mathbf{Z}}_t$, $t=1,\ldots,n$, as in \eqref{standardized}\;
 - Apply NL to the sample $\tilde{\mathbf{Z}}_1, \ldots \tilde{\mathbf{Z}}_n$ to obtain $\LRQIS$\;
 - Calculate $\mathbf{\hat{H}}$ as in \eqref{Hestimate}\; 
\Return{$\mathbf{\hat{H}}$}
 \caption{R-NL($\mathbb{Y}$)}
 \label{RNL}
\end{algorithm}

\begin{algorithm}
\normalsize
\textbf{Inputs}: centered data $\mathbb{Y}$\;
\textbf{Output}: $\mathbf{\hat{H}}$\;
- Calculate the diagonal matrix of sample standard deviations $\boldsymbol{\hat{\sigma}}$ of $\mathbb{Y}$\;
- Calculate $\mathbb{X}\defeq \{\mathbf{X}_1, \ldots, \mathbf{X}_T\}$ as in \eqref{firststand}\;
- Obtain $\mathbf{\hat{H}}_0\defeq$ R-NL($\mathbb{X}$) using Algorithm \ref{RNL}\;
 - Calculate $\mathbf{\hat{H}}\defeq p \cdot \boldsymbol{\hat{\sigma}}\mathbf{\hat{H}}_0 \boldsymbol{\hat{\sigma}}/\Tr(\boldsymbol{\hat{\sigma}}\mathbf{\hat{H}}_0 \boldsymbol{\hat{\sigma}})$\; 
\Return{$\mathbf{\hat{H}}$}
 \caption{R-C-NL($\mathbb{Y}$)}
 \label{RCNL}
\end{algorithm}

\subsection{Robust Correlation-Based Nonlinear Shrinkage}
In the context of covariance matrix estimation,
an alternative approach is to use shrinkage estimation for the
correlation matrix and to estimate the vector of variances separately,
after which one combines the two estimators to obtain a `final'
estimator of the covariance matrix itself.
Such an approach is used by \cite{comfortNL} in a static setting
(that~is, for i.i.d.\ data) and by 
\cite{engle:ledoit:wolf:2019,denard:engle:ledoit:wolf:2022} in a
dynamic setting (that~is, for time series data).
It turns out that by adapting this approach for our method, a considerable boost in performance can be achieved in some settings.

In particular, we first calculate the sample variances $\hat{\sigma}_1^2, \ldots, \hat{\sigma}_p^2$ and obtain the scaled data as
\begin{align}\label{firststand}
    \mathbf{X}_t\defeq \boldsymbol{\hat{\sigma}}^{-1} \mathbf{Y}_t~,
\end{align}
where $\boldsymbol{\hat{\sigma}} \defeq \diag(\hat{\sigma}_1, \ldots, \hat{\sigma}_p)$.
Then \mbox{R-NL} is applied to $\mathbf{Z}_t \defeq
\boldsymbol{\hat{\sigma}}^{-1}
\mathbf{Y}_t/\|\boldsymbol{\hat{\sigma}}^{-1} \mathbf{Y}_t \|$ to
obtain $\hat{\mathbf{H}}_0\defeq p \Vsol
\LRQIS\Vsol^{\top}/\Tr(\LRQIS)$. From these two inputs,
we calculate the `final' estimator of~$\mathbf{H}$~as
\begin{align}\label{Hestimatewithsig}
    \mathbf{\hat{H}}\defeq p\cdot\boldsymbol{\hat{\sigma}} \hat{\mathbf{H}}_0 \boldsymbol{\hat{\sigma}}/\Tr(\boldsymbol{\hat{\sigma}} \hat{\mathbf{H}}_0 \boldsymbol{\hat{\sigma}})~. 
\end{align}
The approach is called ``R-C-NL" and is summarized in Algorithm \ref{RCNL}.

As we demonstrate in Section \ref{sec:mc} this `variation' on our
methodology can have a 
substantial (beneficial)  effect on the performance of the estimator.
However, a potential disadvantage of this approach is that \mbox{R-C-NL} is
no longer rotation-equivariant.

\section{Simulation Study}
\label{sec:mc}

We compare our proposed two methods to several competitors in various simulation scenarios. From the collection of approaches that use Tyler's method together with (linear) shrinkage, we tried to pick the ones most appropriate for our analysis, without handpicking them to showcase the performance of our method. In particular, we did not pick methods that require the choice of a tuning parameter (such as \cite{existence_uniqueness_algorithms, Shrinkage_2017, convexpenalties}) or require $p < n$ such as \cite{Shrinkage_2021}. This leads us to the following benchmarks: 
\begin{itemize}
    \item LS: the linear shrinkage estimator of \cite{Wolflinear}.
    \item NL: the quadratic inverse shrinkage (QIS) estimator of \cite{QIS2020}.
    \item R-LS: the robust linear shrinkage estimator of \cite{linearshrinkageinheavytails}. An estimator that is ``widely used and performs well in practice'' \cite{existence_uniqueness_algorithms}.
    \item R-GMV-LS: the robust linear shrinkage estimator of \cite{yang2014minimum}. This estimator is designed for global-minimum-variance portfolios. 
    \item R-A-LS: the regularized estimator of \cite{zhang2016automatic}, a variation of the robust linear shrinkage approach. 
    \item R-TH: the robust estimator of \cite{thresholdingandTyler}, based on thresholding Tyler's M-estimator.
\end{itemize}

\noindent
Moreover, we consider the following six structures for the true dispersion matrix $\mathbf{H}$:
\begin{itemize}
    \item \sI\  Identity: The identity matrix $\mathbf{I}$.
    \item \sA\  AR: the $(i,j)$ element of $\mathbf{H}$ is $0.7^{\mid i - j\mid}$, as in \cite{linearshrinkageinheavytails}.
    \item \sF\  Full matrix: $1$ on the diagonal and $0.5$ on the off-diagonal.
    \item \sIsig\  Base: diagonal matrix, where $20\%$ of the 
      diagonal elements are equal to $1$, $40\%$ of the diagonal elements are
      equal to $3$, and $40\%$ of the diagonal elements are equal to
      $10$, as in \cite{WolfNL, Analytical_Shrinkage, QIS2020}.
    \item \sAsig\  AR (non-constant diag): start with  $\mathbf{H}$ as
      in (A) and then pre- and post-multiply with the square-root of
      the diagonal matrix as in (I$'$).
    \item \sFsig\  Full Matrix (non-constant diag): start with  $\mathbf{H}$ as
      in (F) and then pre- and post-multiply with the square-root of
      the diagonal matrix as in (I$'$).
\end{itemize}

Together these settings cover a wide range of structures for the
dispersion matrix $\mathbf{H}$, from sparse to a ``full'' matrix with
nonzero elements everywhere. Most papers related to our method
simulate from a multivariate $t$-distribution with $3$ or $4$ degrees
of freedom. In contrast, we let the degrees of freedom vary on a grid
from $3$ to ``infinity'', that is, to~the \mbox{Gaussian} case. It appears the
actual sample size does not matter as much as the concentration ratio
in the relative performance of the methods. As such we choose two
concentration ratios, $2/3$ and $4/3$, with a fixed sample size of
$n=300$. In Appendix \ref{sec:further_results} the same analysis
is done for $n=150$ and $p \in \{100, 200\}$. In addition, it contains an analysis comparing \mbox{R-NL} to the methodology of \cite{ourclosestcompetitor} over various shrinkage parameters $\alpha$ and using $\boldsymbol{\Lambda}_0$ as target eigenvalues.

Finally, as a measure for comparing the different methods we consider
the Percentage Relative Improvement in Average Loss (PRIAL) defined as
$$\text{PRIAL}(\hat{\mathbf{H}}_*) \defeq 100 \times \left( 1-
  \frac{\E[\Vert\hat{\mathbf{H}}_* -
    \mathbf{H}\Vert^2]}{\E[\Vert\Sample - \mathbf{H}\Vert^2]}\right)
\%~,$$
where $\hat{\mathbf{H}}_*$ denotes a generic estimator of
$\mathbf{H}$. Note that our definition of PRIAL differs from the one
in \cite{QIS2020}: For both definitions, the value of 0 corresponds to
the (scaled) sample covariance matrix; but the value of 100 corresponds to the
true matrix in our definition whereas it corresponds to the `oracle'
estimator in the class of \cite{stein:1975,stein:1986} in the
definition of \cite{QIS2020}. It does not make sense for us to use the
definition of  \cite{QIS2020}, since our estimator, unlike the their QIS
estimator, is not in the Steinian class. We also note that all matrices in the above PRIAL
are scaled to have trace $p$.

Figures \ref{fig:SimN200} and \ref{fig:SimN400} show the results. It
is immediately visible that in the majority of considered settings
both \mbox{R-NL} and \mbox{R-C-NL} outperform the other estimators. One major exception is the AR case, where the tresholding algorithm R-TH dominates all other methods. Also, in the setting \sI, for both $p=200$ and $p=400$, LS and R-A-LS recognize that shrinking maximally towards a multiple of the identity matrix is optimal, reaching a PRIAL of almost $100\%$ through all $\nu$. Although \mbox{R-NL} and \mbox{R-C-NL} are close, they cannot quite match this strong performance. For the cases \sF \ and \sFsig, when $p=200$, there is a performance drop of our methods compared to LS and NL for $\nu \geq 30$. However, the values of the methods again stay close. As one would expect, the PRIAL values of \mbox{R-NL} and \mbox{R-C-NL} are similar in the first row, where the true matrix~$\mathbf{H}$ has constant diagonal elements. Moreover, although NL is greatly improved upon with both \mbox{R-NL} and \mbox{R-C-NL} for small to moderate $\nu$, both converge to the performance of NL as the degrees of freedom increase. In the case \sIsig, where the diagonal elements are non-constant, \mbox{R-C-NL} attains a strong boost compared to \mbox{R-NL} and NL, such that it outperforms all other benchmarks by a considerable margin. The improvement is smaller, but consistent, for \sAsig \ for both $p$ and for \sFsig \ in the case $p=400$. We note that the consistently high performance of both methods through most settings is quite remarkable. In Appendix \ref{sec:further_results} further simulations with similar findings are presented. Given that other methods, such as that of \cite{zhang2016automatic}, perform quite
well on balance, but can collapse in some cases, this consistent-throughout
performance appears remarkable.

It is also worth noting how well the linear shrinkage methods perform
in the setting \sIsig\ with non-constant diagonal elements. This may seem
counterintuitive, since the shrinkage target is a multiple of the
identity matrix. Indeed, the diagonal elements of the linear shrinkage
estimates are close to constant in these cases. However, at least for
heavy-tails, the errors the sample matrix admits on the off-diagonal
elements far outweigh the errors of constant diagonal
elements. Additionally, LS, to which most other papers compare their
methods, is extremely competitive with the robust methods (even if
$\nu$ is very small). This is especially true for $\mathbf{H}$ with
non-constant diagonal elements, which most of the previous papers do
not consider. The good performance of LS might also be due to the
relatively high sample size used in this paper compared to others.

\begin{figure*}[!ht]
\begin{adjustwidth}{-2cm}{}
    \centering
    \includegraphics[width=20cm]{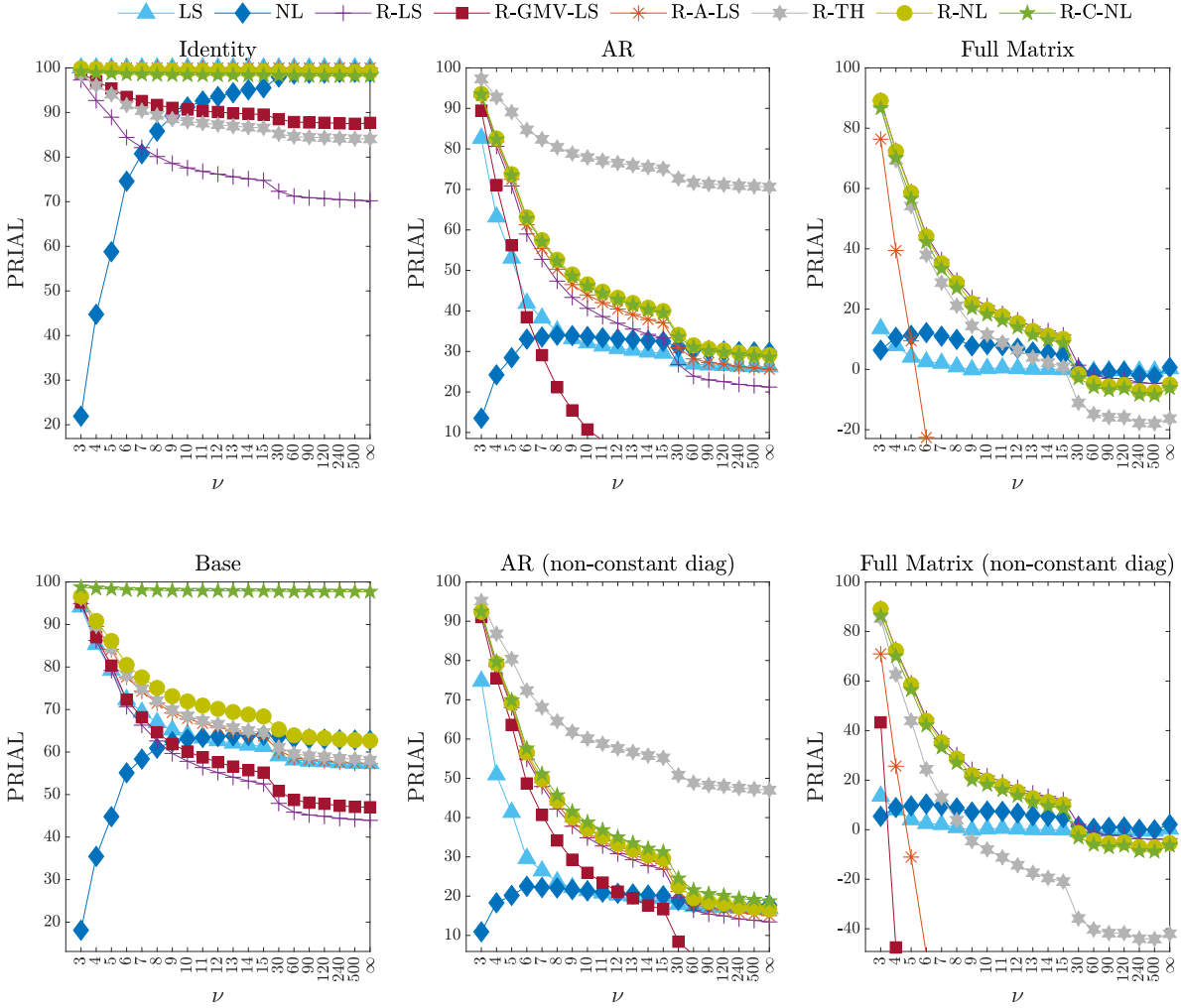}
    \end{adjustwidth}
    \vspace*{-0cm}
    \caption{Percentage Relative Improvement in Average Loss (PRIAL) for various dispersion matrix structures, $\nu \in \{3,4,\dots,14,15,30,60,90,120,240,500,\infty\}$, $p=200$ and $n=300$. The plots are scaled such that they extend to $100$ and NL and R-C-NL are fully visible.}
    \label{fig:SimN200}
\end{figure*}

\begin{figure*}[!ht]
\begin{adjustwidth}{-2cm}{}
    \centering
    \includegraphics[width=20cm]{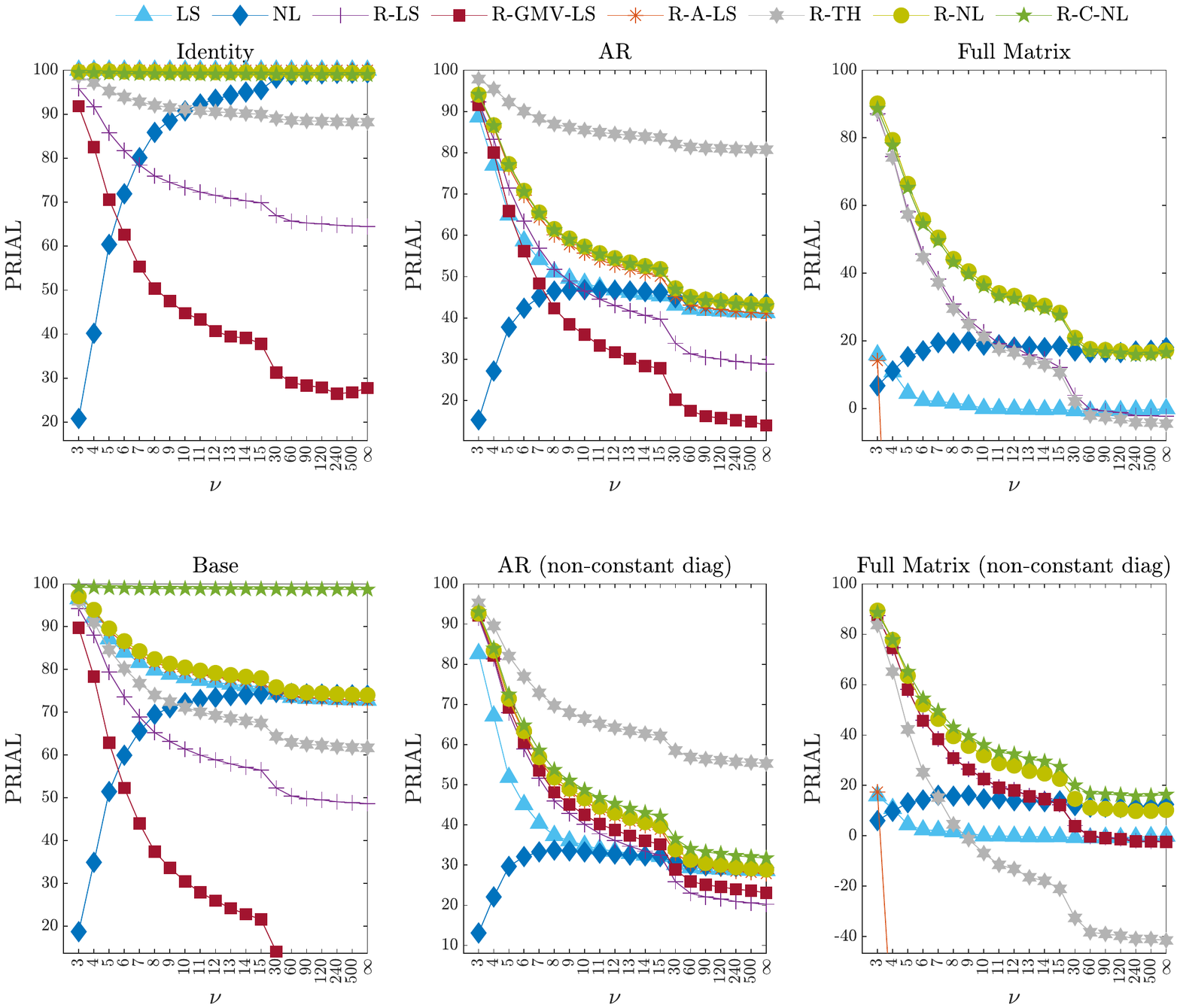}
    \end{adjustwidth}
    \vspace*{0cm}
    \caption{Percentage Relative Improvement in Average Loss (PRIAL) for various dispersion matrix structures, $\nu \in \{3,4,\dots,14,15,30,60,90,120,240,500,\infty\}$, $p=400$ and $n=300$. The plots are scaled such that they extend to $100$ and NL and R-C-NL are fully visible.}
    \label{fig:SimN400}
\end{figure*}

\section{Empirical Study}
\label{sec:empirics}

From the Center for Research in Security Prices (CRSP) we download daily simple percentage returns of the NYSE, AMEX, and NASDAQ stock exchanges, see \url{https://www.crsp.org/node/1/activetab\%3Ddocs} for a documentation. The historical data ranges from 02.01.1976 until 31.12.2020 and contains $p^* = 23'131$ stocks in total. The outline of the empirical section is inspired by \cite{Gian2019}.

We conduct a rolling window type exercise, where we consider an estimation window of one year (252 days) and another one with five years (1260 days). In each rolling window we estimate the covariance matrix and perform a minimum variance portfolio optimization with no short selling limits. In the unconstrained case, the global minimum variance portfolio problem is formulated as 

$$\min\limits_{\tilde{w}} \tilde{w}^{\top} \boldsymbol{\hat{\Sigma}} \tilde{w}$$
$$\text{subject to} \quad \tilde{w}^{\top}\mathbbm{1} = 1~,$$
where $\mathbbm{1}$ denotes a conformable vector of ones. In the absence of any short-sales constraints the problem has the analytical solution 
$$w \defeq \frac{\boldsymbol{\hat{\Sigma}}^{-1}\mathbbm{1}}{\mathbbm{1}^{\top}\boldsymbol{\hat{\Sigma}}^{-1}\mathbbm{1}}~.$$

The resulting number of shares (rather than portfolio weights) are
then kept fixed for the following
21~days (out-of-sample window). Afterwards the rolling window moves
forward by 21~days, the covariance matrix is re-estimated and the
weights are updated accordingly. In short, we rebalance the portfolio
once a `month', and there are no transaction costs during the `month',
where our definition of a `month' corresponds to 21
consecutive trading days rather than a calendar month.
Depending on the size of the estimation window, the
out-of-sample period starts on 14-Jan-1977 respectively 13-Jan-1981.
This results in 528 respectively 480 out-of-sample months. In
each month we only consider stocks which have no more than 32 days of
missing values during the estimation window and a complete return in
the out-of-sample window. The missing values in the remaining universe
are set to $0$. Further, every month, only the $p$ stocks with the
highest market capitalization are considered, where
$p~\in~\{100,500,1000\}$.

The solution of the minimum variance portfolio only depends on the
second moment, that is, the covariance matrix. Therefore, as a portfolio
evaluation criterion, the out-of-sample standard deviation is the
leading criterion of interest. Hence, in the main text we consider the following two portfolio performance measures:
\begin{itemize}
    \item SD: annualized standard deviation of portfolio returns.
    \item TO: average monthly turnover given by 
    
    $$\text{TO} \defeq
    \frac{1}{(\tau-1)}\sum_{h=1}^{\tau-1}\sum_{j=1}^{p^{*}}\mid
    w_{j,h+1} - w_{j,h}^\text{hold}\mid~,$$
Here $p^*$ denotes the size of the `combined' investment universe over
both months, $h$ and $h+1$; in general some stocks leave the universe,
while the same number of new stocks enter the universe, as one advances
from month $h$ to month $h+1$ such that $p^* \ge p$. Furthermore,
$\tau$ denotes the number of out-of-sample months (528 and 480, respectively) and 
    
    $$w_{j,h}^\text{hold} \defeq \frac{w_{j,h}\alpha_{j,h}}{\sum_{j=1}^{p^{*}}w_{j,h}\alpha_{j,h}}~,$$
    
    with $$\alpha_{j,h} \defeq \prod_{s=0}^{20}(1+r_{j,t_h+s})$$ representing the return evolution in the days of month $h$. Note: if stock $j$ is not contained in the universe during month $h$, we set $w_{j,h} = 0$ and $r_{j,t_h+s} = 0 \ \forall \ s\in\{0,1,\dots,20\}$. 
\end{itemize}

\noindent
We consider the same competitors as in Section \ref{sec:mc}, with one exception: We remove the thresholding method ``R-TH'', since the matrix inverse necessary for the portfolio optimization cannot always be computed. We also add the (scaled) sample covariance matrix $\Sample$, which we denote with $S$.

Additionally, in order to test whether the difference of the out-of-sample standard deviation between NL and \mbox{R-C-NL} is significantly different from zero, we apply the HAC inference of \cite{ledoit2011robust}.

Table~\ref{table:Tab:Performance_CRSP} presents the main results; for
additional results pertaining to alternative performance measures, 
see Appendix \ref{sec:further_results}.
The findings are as follows:
\begin{itemize}
    \item \mbox{R-C-NL} has the lowest SD in every scenario.
    \item \mbox{R-C-NL} has a significantly lower SD than NL in every
    scenario.
    \item \mbox{R-NL} has a lower SD than NL, except for $p=1000$, where they have a similar SD.
    \item The difference in SD between \mbox{R-NL} and \mbox{R-C-NL} increases as $p$ increases.
    \item For both $n=252$ and $n=1260$, the robust estimators have a
      lower SD than the non-robust estimators when $p$ is small. For
      large $p$ the non-robust estimators perform similarly or better
      than the robust estimators. This holds for linear shrinkage and
      nonlinear shrinkage estimators; an exception is R-GMV-LS  in the case $n=1260$.
    \item For $p=500$ and $p=1000$, NL always outperforms the robust linear shrinkage estimators.
    \item For $p=100$, NL and the robust linear shrinkage estimators perform similarly.
    \item Except for the case $p=100$, \mbox{R-C-NL} has the lowest TO in
      every scenario.
    \item \mbox{R-NL} always has lower TO than NL.
\end{itemize}

We note that the improvement of \mbox{R-NL} over NL in terms of SD is relatively small, and indeed substantial improvements in SD only occur for \mbox{R-C-NL}. In other words, the strongest improvement appears to stem from using NL on the correlation matrix. However, NL is known to be a very strong benchmark in unconstrained PF optimization. Consequently, already the comparatively small gain of \mbox{R-NL} over NL leads to the lowest SD in five out of six considered scenarios.

\begin{table*}[!ht]                                                  \footnotesize    
\caption{Out-of-sample portfolio statistics for the largest $p$ stocks on CRSP and an estimation window of $n=252$ and $n=1260$ days, respectively. Significant outperformance of the introduced R-C-NL estimator over NL in terms of SD is indicated by asterisks: **~indicates significance at the 0.05 level and ***~indicates significance at the 0.01 level.}
\begin{adjustwidth}{-1.25cm}{}
\begingroup
\setlength{\tabcolsep}{6pt} 
\renewcommand{\arraystretch}{1.25}
\begin{tabular}{L{0.5cm}C{1.8cm}C{1.8cm}C{1.8cm}C{1.8cm}C{1.8cm}C{1.8cm}C{1.8cm}C{1.8cm}}

\hline                                                                                                                 
& S & LS & NL & R-LS & R-GMV-LS & R-A-LS & R-NL & R-C-NL \\                                         
\midrule

\hline \hline                                                                                                                   
 & \multicolumn{8}{c}{\rule[-2.5mm]{0mm}{8mm}\emph{$n=252$}}  \\                   
\hline \hline                                                                               
& \multicolumn{8}{c}{\rule[-2.5mm]{0mm}{8mm}\emph{$p=100$}}  \\                       
\hline                                                                                 
SD & $13.41$ & $12.36$ & $11.91$ & $12.29$ & $11.91$ & $12.20$ & $11.86$ & $\tcb{11.60}\text{***}$ \\  
\hline                                                                                 
TO & $2.23$ & $1.40$ & $1.02$ & $1.38$ & $\tcb{0.79}$ & $1.25$ & $0.91$ & $0.86$ \\          
\hline                                                                                 
& \multicolumn{8}{c}{\rule[-2.5mm]{0mm}{8mm}\emph{$p=500$}}  \\ 
\hline                                                                                 
SD & $-$ & $9.65$ & $9.10$ & $10.01$ & $9.41$ & $9.65$ & $9.08$ & $\tcb{8.49}\text{***}$ \\    
\hline                                                                                 
TO & $-$ & $2.29$ & $1.16$ & $2.72$ & $1.59$ & $2.18$ & $1.04$ & $\tcb{1.03}$ \\       
\hline                                                                                 
& \multicolumn{8}{c}{\rule[-2.5mm]{0mm}{8mm}\emph{$p=1000$}}   \\
\hline                                                                                 
SD & $-$ & $8.33$ & $8.23$ & $8.38$ & $8.45$ & $8.34$ & $8.23$ & $\tcb{7.04}\text{***}$ \\       
\hline                                                                                 
TO & $-$ & $1.59$ & $1.01$ & $1.76$ & $0.93$ & $1.56$ & $0.92$ & $\tcb{0.90}$ \\  
\midrule
\hline \hline                                                                               
& \multicolumn{8}{c}{\rule[-2.5mm]{0mm}{8mm}\emph{$n=1260$}}  \\                                
\hline \hline                                                                                                          
& \multicolumn{8}{c}{\rule[-2.5mm]{0mm}{8mm}\emph{$p=100$}}  \\                                 
\hline                                                                                                                 
SD & $12.83$ & $12.78$ & $12.68$ & $12.67$ & $12.69$ & $12.68$ & $12.59$ & $\tcb{12.55}\text{**}$ \\  
\hline                                                                                 
TO & $0.65$ & $0.60$ & $0.57$ & $0.54$ & $\tcb{0.38}$ & $0.53$ & $0.49$ & $0.48$ \\          
\hline                                                                                 
& \multicolumn{8}{c}{\rule[-2.5mm]{0mm}{8mm}\emph{$p=500$}}  \\ 
\hline                                                                                 
SD & $10.25$ & $9.82$ & $9.35$ & $9.72$ & $9.44$ & $9.69$ & $9.32$ & $\tcb{9.12}\text{***}$ \\
\hline                                                                                 
TO & $1.97$ & $1.53$ & $0.92$ & $1.45$ & $0.72$ & $1.41$ & $0.79$ & $\tcb{0.67}$ \\          
\hline                                                                                 
& \multicolumn{8}{c}{\rule[-2.5mm]{0mm}{8mm}\emph{$p=1000$}}   \\
\hline                                                                                 
SD & $12.55$ & $8.91$ & $8.04$ & $9.04$ & $8.12$ & $8.91$ & $8.05$ & $\tcb{7.46}\text{***}$ \\ 
\hline                                                                                 
TO & $6.56$ & $2.38$ & $0.94$ & $2.52$ & $0.84$ & $2.35$ & $0.80$ & $\tcb{0.64}$ \\ 
\hline                                                                                                                                                                                          
\end{tabular} 
\endgroup                                                            
\label{table:Tab:Performance_CRSP}    
\end{adjustwidth}
\end{table*} 

\newpage

\normalsize

\section{Conclusion}
\label{sec:conclusion}

This paper combines nonlinear shrinkage with Tyler's method, thereby
creating a fast and stable algorithm to estimate the dispersion matrix
in elliptical models; the resulting estimator is robust against both
heavy tails and high dimensions.
We developed the algorithm by separating
calculation of the eigenvalues and the eigenvectors and showed that
eigenvectors could be obtained by an iterative procedure. We also
showed that the resulting \mbox{R-NL} estimator is still
rotation-equivariant, although it no longer is contained in the
Steinian class of rotation-equivariant estimators that keeps the
vectors of the sample covariance matrix and only shrinks the sample
eigenvalues. We also compared our approach to existing methods from
the literature using both extensive simulations and an application to
real data, showcasing its favorable performance.
Last but not least, it turns out that a further performance boost can
be obtained by using our method on scaled data, which basically
amounts to separating the problem of estimating a covariance matrix
into estimation of individual variances and estimation of the
correlation matrix; the resulting estimator is called \mbox{R-C-NL}.



\newpage

\appendix

\section{Further Empirical and Simulation Results}
\label{sec:further_results}

Figures \ref{fig:SimN100T150} and \ref{fig:SimN200T150} show the simulation results for $n=150$ and $p=100$ and $p=200$ respectively. Figure \ref{fig:SimN200T300_alpha} compares with the ``SRTy'' estimator of \cite{ourclosestcompetitor}, by using the NL shrinkage eigenvalues as target eigenvalues. That is, the target matrix of eigenvalues is given as $\boldsymbol{\Lambda}_0$ as in the main paper. The shrinkage strength $\beta$ is given by
\[
\beta \defeq \frac{\alpha}{1+\alpha},
\]
in \cite{ourclosestcompetitor}, so that $\beta=0$ corresponds to no shrinkage, whereas $\beta=1$ corresponds to maximal shrinkage (or setting $\alpha=\infty$).

Interestingly, in all but the settings (F) and (F'), setting $\beta=1$, and thus maximally shrinking towards $\boldsymbol{\Lambda}_0$ appears beneficial. In fact, the performance in this settings of $\beta=1$ is about the same as \mbox{R-NL}. However, in the settings (F) and (F'), the situation is reversed and performance gets worse, the higher $\beta$ is chosen. With the additional updating step performed in \mbox{R-NL}, performance is about the same as choosing $\beta=0$, showing that \mbox{R-NL} can have a substantial benefit over using SRTy with $\beta=1$.

\begin{figure*}[!ht]
\begin{adjustwidth}{-2cm}{}
    \centering
    \includegraphics[width=20cm]{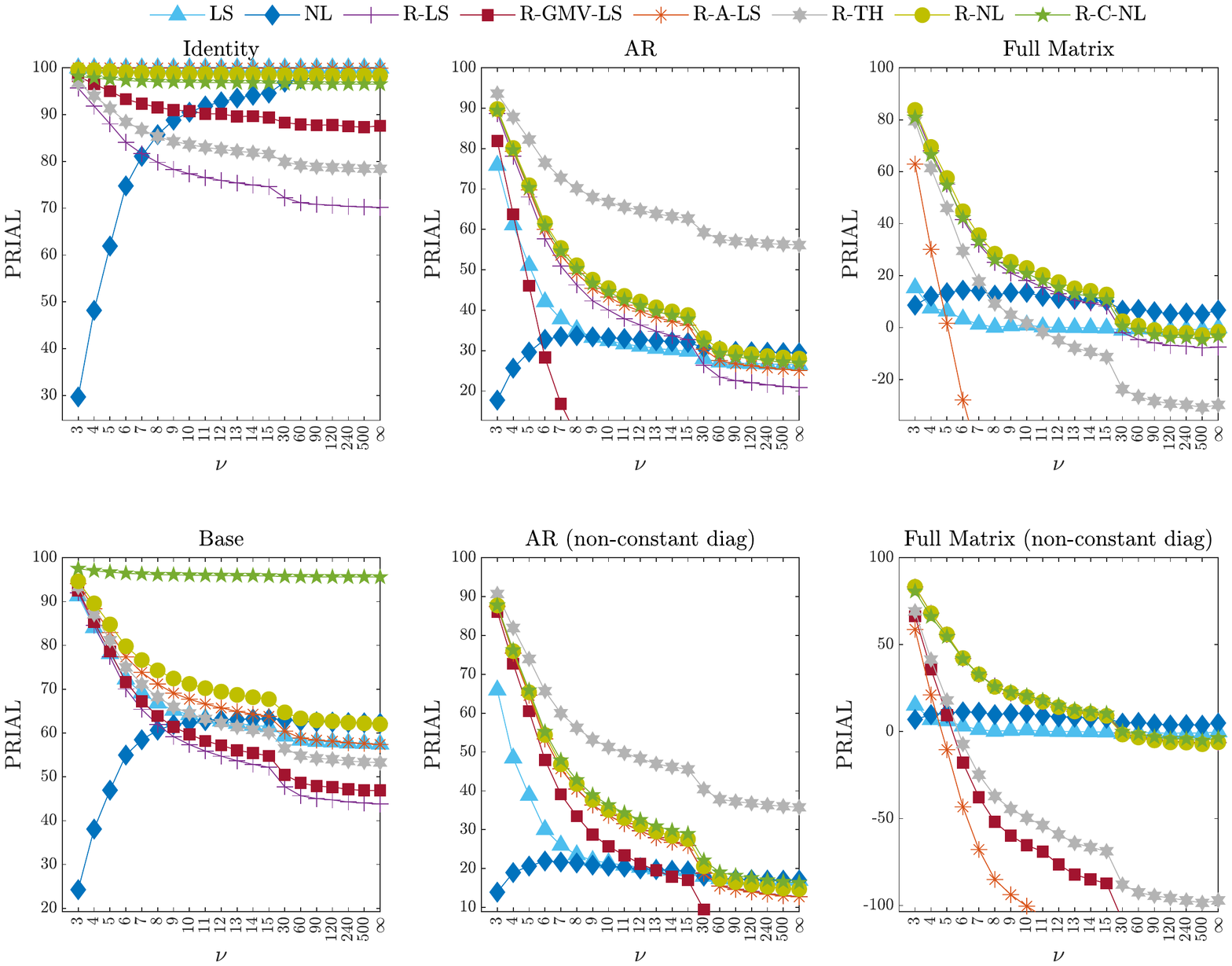}
    \end{adjustwidth}
    \vspace*{-1.5cm}
    \caption{Percentage Relative Improvement in Average Loss (PRIAL) for various dispersion matrix structures, $\nu \in \{3,4,\dots,14,15,30,60,90,120,240,500,\infty\}$, $p=100$ and $n=150$. The plots are scaled such that they extend to $100$ and NL and \mbox{R-C-NL} are fully visible. }
    \label{fig:SimN100T150}
\end{figure*}

\begin{figure*}[!ht]
\begin{adjustwidth}{-2cm}{}
    \centering
    \includegraphics[width=20cm]{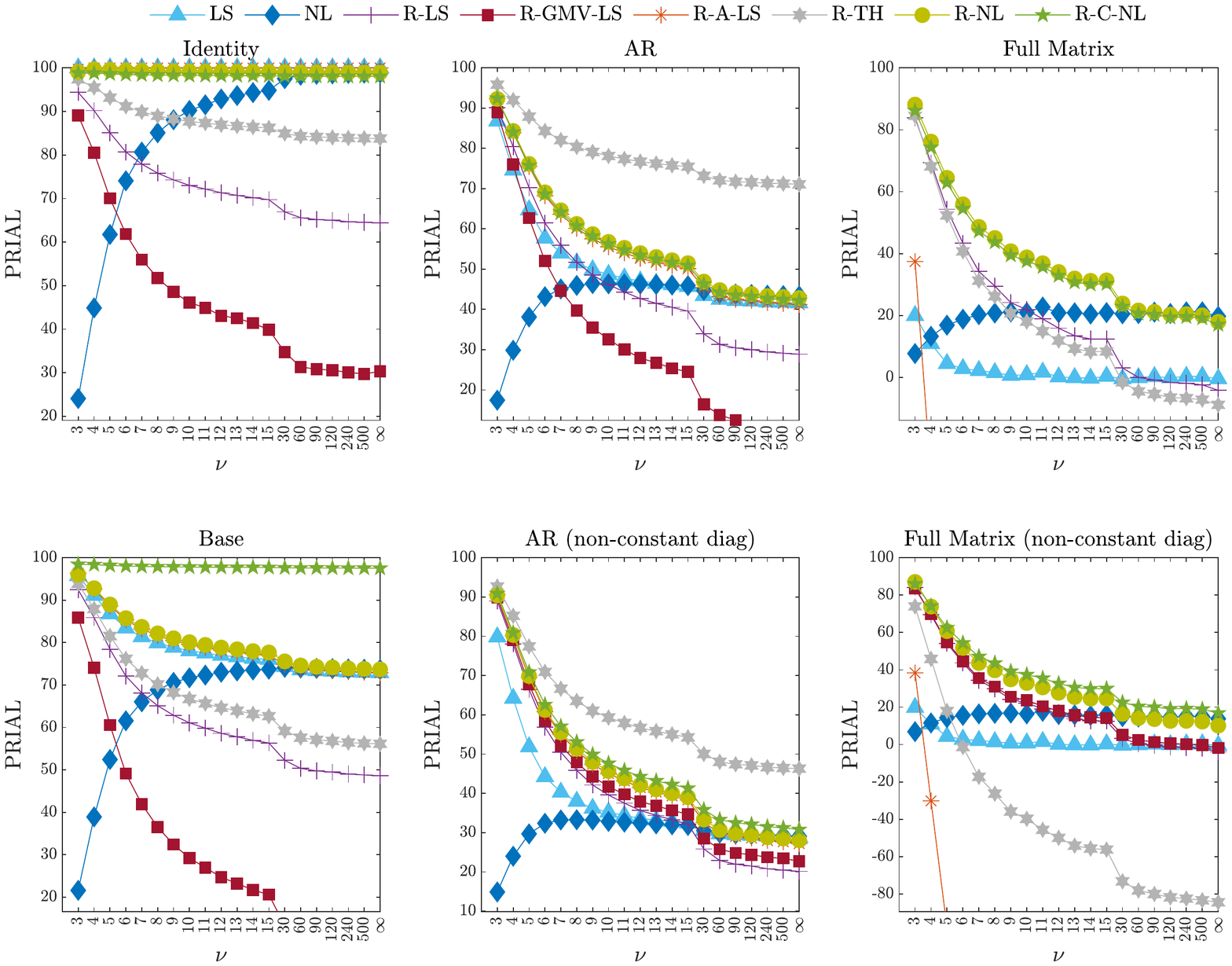}
    \end{adjustwidth}
    \vspace*{-1.5cm}
    \caption{Percentage Relative Improvement in Average Loss (PRIAL) for various dispersion matrix structures, $\nu \in \{3,4,\dots,14,15,30,60,90,120,240,500,\infty\}$, $p=200$ and $n=150$. The plots are scaled such that they extend to $100$ and NL and \mbox{R-C-NL} are fully visible.}
    \label{fig:SimN200T150}
\end{figure*}

\begin{figure*}[!ht]
\begin{adjustwidth}{-2cm}{}
    \centering
    \includegraphics[width=20cm]{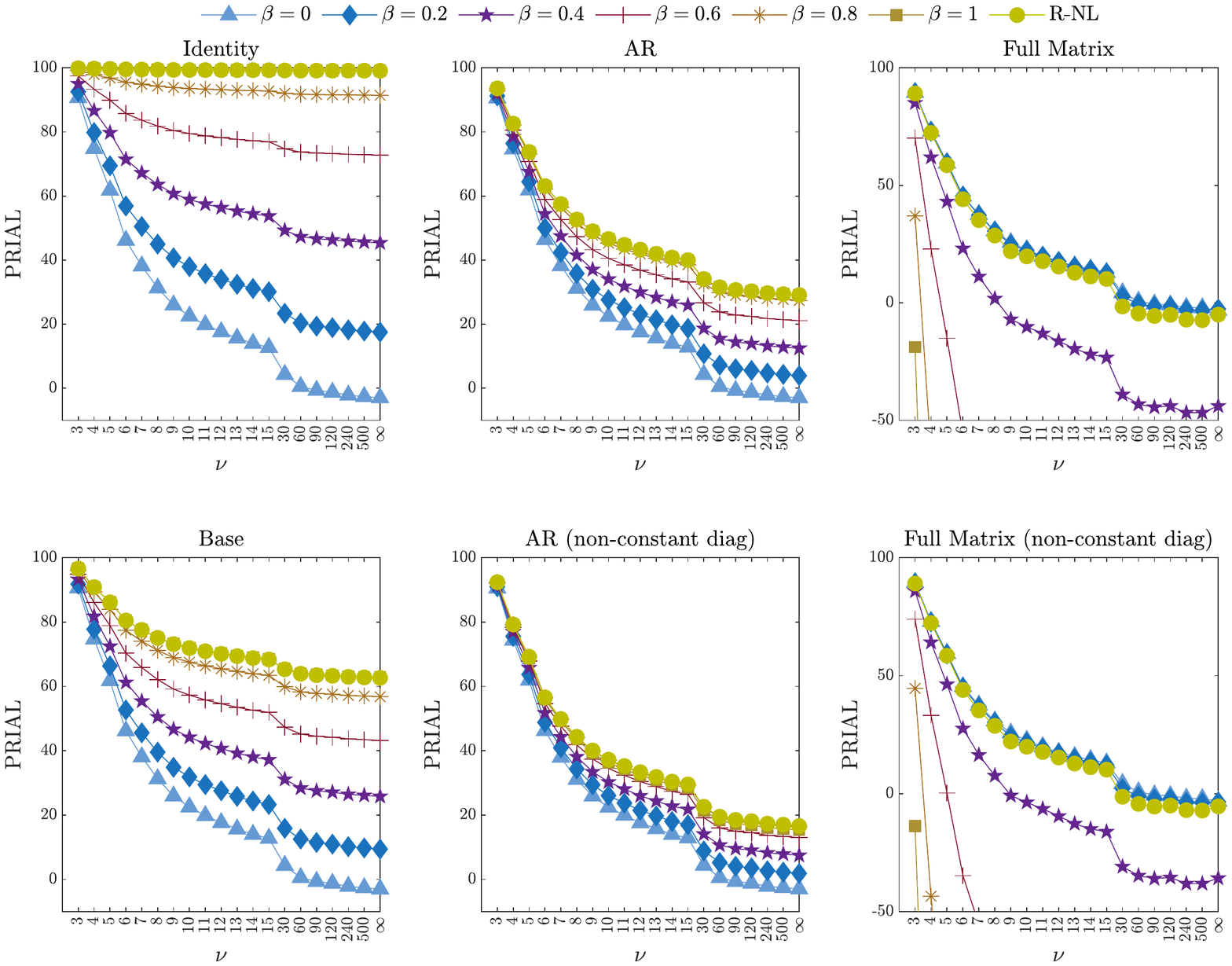}
    \end{adjustwidth}
    \vspace*{-1.5cm}
    \caption{Percentage Relative Improvement in Average Loss (PRIAL) for various dispersion matrix structures, $\nu \in \{3,4,\dots,14,15,30,60,90,120,240,500,\infty\}$, $p=200$ and $n=300$.}
    \label{fig:SimN200T300_alpha}
\end{figure*}

For our empirical application in Section 5 of the main paper, the following four additional performance measures are reported:

\begin{itemize}
    \item AV: annualized average simple percentage portfolio return.
    \item TR: final cumulative simple percentage portfolio return.
    \item MD: percentage maximum drawdown given by $$\text{MD} \defeq \max_{\bar{t}\in (0,n)}\left(\max_{t\in (0,\bar{t})}\left[\frac{\tilde{R}_t-\tilde{R}_{\tilde{t}}}{\tilde{R}_t}\right]\right)~,$$ where $\tilde{R}_t$ is the cumulative simple percentage portfolio return at day $t$. 
    \item IR: annualized information ratio given by $$IR \defeq \frac{AV}{SD}~.$$
\end{itemize}

\begin{table*}[!ht]
\footnotesize 
\caption{Additional out-of-sample portfolio statistics for the daily $p$ largest stocks on CRSP and an estimation window of $n=252$ and $n=1260$ days, respectively.}
\begin{adjustwidth}{-1.25cm}{}
\begingroup
\setlength{\tabcolsep}{6pt} 
\renewcommand{\arraystretch}{1.25}
\begin{tabular}{L{0.5cm}C{1.8cm}C{1.8cm}C{1.8cm}C{1.8cm}C{1.8cm}C{1.8cm}C{1.8cm}C{1.8cm}}                                                            

\hline                                                                                                                 
& S & LS & NL & R-LS & R-GMV-LS & R-A-LS & R-NL & R-C-NL \\                                         
\midrule

\hline \hline                                                                                                                   
 & \multicolumn{8}{c}{\rule[-2.5mm]{0mm}{8mm}\emph{$n=252$}}  \\                   

\hline \hline                                                                                                             
& \multicolumn{8}{c}{\rule[-2.5mm]{0mm}{8mm}\emph{$p=100$}}  \\                       

\hline                                                                                                                   
IR & $0.69$ & $0.80$ & $0.87$ & $0.81$ & $0.87$ & $0.82$ & $0.87$ & $0.91$ \\                                     
\hline                                                                                                            
AV & $9.24$ & $9.94$ & $10.31$ & $9.97$ & $10.37$ & $9.97$ & $10.32$ & $10.55$ \\                                 
\hline                                                                                                            
TR & $3,813.45$ & $5,548.60$ & $6,710.69$ & $5,641.21$ & $6,902.58$ & $5,677.55$ & $6,767.45$ & $7,607.62$ \\             
\hline                                                                                                            
MD & $46.56$ & $35.20$ & $38.48$ & $37.94$ & $33.40$ & $37.05$ & $35.48$ & $33.26$ \\                     
\hline                                                                                                            
& \multicolumn{8}{c}{\rule[-2.5mm]{0mm}{8mm}\emph{$p=500$}}  \\                            
\hline                                                                                                            
IR & $-$ & $1.17$ & $1.21$ & $1.16$ & $1.20$ & $1.20$ & $1.22$ & $1.35$ \\                                    
\hline                                                                                                            
AV & $-$ & $11.27$ & $11.01$ & $11.57$ & $11.26$ & $11.56$ & $11.07$ & $11.50$ \\                          
\hline                                                                                                            
TR & $-$ & $11,493.55$ & $10,447.08$ & $12,921.46$ & $11,521.84$ & $13,053.61$ & $10,776.80$ & $13,303.94$ \\
\hline                                                                                                            
MD & $-$ & $29.41$ & $29.24$ & $29.92$ & $30.33$ & $28.56$ & $30.61$ & $28.22$ \\                        
\hline                                                                                                            
& \multicolumn{8}{c}{\rule[-2.5mm]{0mm}{8mm}\emph{$p=1000$}}  \\                           
\hline                                                                                                            
IR & $-$ & $1.47$ & $1.44$ & $1.46$ & $1.42$ & $1.46$ & $1.43$ & $1.74$ \\                                     
\hline                                                                                                            
AV & $-$ & $12.21$ & $11.87$ & $12.23$ & $12.04$ & $12.21$ & $11.76$ & $12.22$ \\                            
\hline                                                                                                            
TR & $-$ & $18,342.65$ & $15,877.34$ & $18,502.15$ & $16,942.34$ & $18,347.06$ & $15,080.18$ & $19,254.29$ \\   
\hline                                                                                                            
MD & $-$ & $33.88$ & $34.39$ & $34.52$ & $34.82$ & $33.63$ & $35.19$ & $26.75$ \\                         
            
\midrule

\hline \hline                                                                                                                         
& \multicolumn{8}{c}{\rule[-2.5mm]{0mm}{8mm}\emph{$n=1260$}}   \\                                      
\hline \hline                                                                                                                         
& \multicolumn{8}{c}{\rule[-2.5mm]{0mm}{8mm}\emph{$p=100$}}   \\                                       
\hline                                                                                                                       
IR & $0.93$ & $0.94$ & $0.94$ & $0.94$ & $0.95$ & $0.94$ & $0.93$ & $0.95$ \\                                     
\hline                                                                                                            
AV & $11.95$ & $12.06$ & $11.97$ & $11.89$ & $12.01$ & $11.90$ & $11.74$ & $11.86$ \\                             
\hline                                                                                                            
TR & $8,428.17$ & $8,863.69$ & $8,585.57$ & $8,323.75$ & $8,709.89$ & $8,341.21$ & $7,862.18$ & $8,267.73$ \\             
\hline                                                                                                            
MD & $39.08$ & $37.83$ & $38.08$ & $36.21$ & $34.00$ & $36.12$ & $35.31$ & $34.98$ \\                     
\hline                                                                                                            
& \multicolumn{8}{c}{\rule[-2.5mm]{0mm}{8mm}\emph{$p=500$}} \\                            
\hline                                                                                                            
IR & $1.08$ & $1.14$ & $1.21$ & $1.16$ & $1.24$ & $1.17$ & $1.23$ & $1.30$ \\                                     
\hline                                                                                                            
AV & $11.04$ & $11.23$ & $11.33$ & $11.23$ & $11.67$ & $11.30$ & $11.47$ & $11.82$ \\                             
\hline                                                                                                            
TR & $6,599.24$ & $7,256.96$ & $7,703.75$ & $7,271.52$ & $8,779.64$ & $7,487.65$ & $8,162.71$ & $9,462.95$ \\             
\hline                                                                                                            
MD & $33.16$ & $31.64$ & $32.92$ & $30.85$ & $31.07$ & $30.79$ & $31.48$ & $31.56$ \\                     
\hline                                                                                                            
& \multicolumn{8}{c}{\rule[-2.5mm]{0mm}{8mm}\emph{$p=1000$}} \\                           
\hline                                                                                                            
IR & $0.93$ & $1.37$ & $1.52$ & $1.37$ & $1.53$ & $1.39$ & $1.53$ & $1.67$ \\                                     
\hline                                                                                                            
AV & $11.71$ & $12.23$ & $12.26$ & $12.38$ & $12.46$ & $12.38$ & $12.29$ & $12.50$ \\                             
\hline                                                                                                            
TR & $7,799.17$ & $11,257.77$ & $11,731.10$ & $11,893.95$ & $12,676.34$ & $11,951.51$ & $11,857.04$ & $13,137.98$ \\      
\hline                                                                                                            
MD & $37.38$ & $33.98$ & $32.18$ & $33.42$ & $32.07$ & $33.58$ & $32.28$ & $27.65$ \\              
\hline                                                                                                                                                                                           
\end{tabular} 
\endgroup                                                            
\label{table:Tab:Performance_CRSP_appendix}    
\end{adjustwidth}
\end{table*} 


\clearpage
 \newpage

 \bibliographystyle{apalike}
 \bibliography{multt}

\end{document}